\documentclass[12pt]{article}

\usepackage{mathpazo}

\usepackage{hyperref}        \hypersetup{colorlinks,allcolors=blue,citecolor=blue}   
\usepackage{color}
\usepackage{graphicx,amsmath,amssymb,amsfonts,bm,epsfig,epsf,url,dsfont}

\usepackage{amsthm,mathrsfs}
\usepackage{tikz}
\usepackage{bbm}      
\usepackage{booktabs}
\usepackage{cases}
\usepackage{fullpage}
\usepackage[small,bf]{caption}
\usepackage[top=1in,bottom=1in,left=1in,right=1in]{geometry}
\usepackage{fancybox}
\usepackage{algorithm}
\usepackage{ver batim}
\usepackage{algorithmic}
\usepackage{amsthm}
\usepackage{mathtools,cite}
\usepackage{scalerel,stackengine}
\stackMath
\newcommand\reallywidehat[1]{%
\savestack{\tmpbox}{\stretchto{%
  \scaleto{%
    \scalerel*[\widthof{\ensuremath{#1}}]{\kern-.6pt\bigwedge\kern-.6pt}%
    {\rule[-\textheight/2]{1ex}{\textheight}}
  }{\textheight}%
}{0.5ex}}%
\stackon[1pt]{#1}{\tmpbox}%
}

\usepackage{caption}
\usepackage{subcaption}

\usepackage{silence}
\newcommand{\bE}{\mathbb{E}}

\newcommand{\Var}{\text{Var}}

\newtheorem{definition}{Definition}

\newtheorem{theorem}{Theorem}
\newtheorem{corollary}{Corollary}
\newtheorem{conjecture}{Conjecture}
\newtheorem{lemma}{Lemma}

\newtheorem{problem}{Problem}

\newtheorem{rmk}{Remark}

\newenvironment{fminipage}%
  {\begin{Sbox}\begin{minipage}}%
  {\end{minipage}\end{Sbox}\fbox{\TheSbox}}

\makeatletter
\newcommand*{\rom}[1]{\expandafter\@slowromancap\romannumeral #1@}
\makeatother

\newcommand{\Ind}{\mathbbm{1}}

\newcommand{\abs}[1]{\left|#1\right|}

\newcommand {\pr} {\mathbb{P}}

\newcommand{\calA}{{\cal A}}
\newcommand{\calB}{{\cal B}}
\newcommand{\calC}{{\cal C}}

\newcommand{\calE}{{\cal E}}

\newcommand{\calG}{{\cal G}}
\newcommand{\calH}{{\cal H}}
\newcommand{\calI}{{\cal I}}

\newcommand{\calK}{{\cal K}}
\newcommand{\calL}{{\cal L}}

\newcommand{\calQ}{{\cal Q}}
\newcommand{\calR}{{\cal R}}
\newcommand{\calS}{{\cal S}}

\newcommand{\calV}{{\cal V}}
\newcommand{\calW}{{\cal W}}

\newcommand{\be}{\begin{equation}}
\newcommand{\ee}{\end{equation}}
\newcommand{\beqna}{\begin{eqnarray}}
\newcommand{\eeqna}{\end{eqnarray}}

\newcommand{\p}[1]{\left(#1\right)}
\newcommand{\pp}[1]{\left[#1\right]}
\newcommand{\ppp}[1]{\left\{#1\right\}}
\newcommand{\norm}[1]{\left\|#1\right\|}

\newcommand{\indep}{\perp \!\!\! \perp}
\usepackage{xspace}

\allowdisplaybreaks

\newcommand{\s}[1]{\mathsf{#1}}

\makeatletter
\def\thanks#1{\protected@xdef\@thanks{\@thanks
        \protect\footnotetext{#1}}}
\makeatother

\begin{document}
\title{Planted Bipartite Graph Detection}
\author{Asaf~Rotenberg~~~~~~~~~~~Wasim~Huleihel~~~~~~~~~~~Ofer Shayevitz\thanks{A. Rotenberg, W. Huleihel and O. Shayevitz are with the Department of Electrical Engineering-Systems at Tel Aviv university, {T}el {A}viv 6997801, Israel (e-mail:  \texttt{asaf.rotenberg@gmail.com, wasimh@tauex.tau.ac.il, ofersha@tauex.tau.ac.il}). The work of A. Rotenberg was supported by the ISRAEL SCIENCE FOUNDATION (grants No. 1734/21 and 1766/22). The work of W. Huleihel was supported by the ISRAEL SCIENCE FOUNDATION (grants No. 1734/21). The work of O. Shayevitz was supported by the ISRAEL SCIENCE FOUNDATION (grant No. 1766/22). This paper was presented in part at the 2023 IEEE International Symposium on Information Theory, Taipei, Taiwan.}}
\maketitle

\begin{abstract}

We consider the task of detecting a hidden bipartite subgraph in a given random graph. This is formulated as a hypothesis testing problem, under the null hypothesis, the graph is a realization of an Erd\H{o}s-R\'{e}nyi random graph over $n$ vertices with edge density $q$. Under the alternative, there exists a planted $k_{\s{R}} \times k_{\s{L}}$ bipartite subgraph with edge density $p>q$. We characterize the statistical and computational barriers for this problem. Specifically, we derive information-theoretic lower bounds, and design and analyze optimal algorithms matching those bounds, in both the dense regime, where $p,q = \Theta\left(1\right)$, and the sparse regime where $p,q = \Theta\left(n^{-\alpha}\right), \alpha \in \left(0,2\right]$. We also consider the problem of testing in polynomial-time. As is customary in similar structured high-dimensional problems, our model undergoes an ``easy-hard-impossible" phase transition and computational constraints penalize the statistical performance. To provide an evidence for this statistical computational gap, we prove computational lower bounds based on the low-degree conjecture, and show that the class of low-degree polynomials algorithms fail in the conjecturally hard region.
\end{abstract}
\sloppy
\section{Introduction}
\label{sec:intro}
The problem of detecting hidden structures in large graphs arises frequently in a vast verity of fields, such as, social networks analysis, computational biology, and computer vision, and have been well-studied in the literature in many interesting settings. It is quite clear by now, that the inference of such structures, especially in the high-dimensional regime, impose many statistical and computational challenges \cite{Cai2018}. In this work, we study a hypothesis testing problem that underlies some of these challenges. To motivate this problem, at least mathematically, let us briefly review a few classical results. To that end, consider the following somewhat general formulation of the hidden subgraph detection problem. Observing a graph over $n$ vertices, assume our goal is to infer whether there is a hidden planted structure in the graph, or if it is purely random, say, a realization of an Erd\H{o}s-R\'{e}nyi random graph. Throughout the paper we denote the latter by $\calG(n,q)$, namely, a random graph over $n$ vertices with edge probability $q$. If the planted structure is a fully connected subset of $k<n$ vertices, namely, a clique, then the above problem is well-known as the \emph{planted clique} ($\mathsf{PC}$) problem (e.g., \cite{Jerrum1992}, \cite{Kucera1995}, \cite{Alon1998}, \cite{Barak2016}, \cite{Dekel2014}, \cite{Feldman2017}). If the planted structure is a realization of an Erd\H{o}s-R\'{e}nyi $\calG(k,p)$ random graph, then we get the well-studied \emph{planted dense subgraph} ($\mathsf{PDS}$) problem (e.g., \cite{Hajek2015}, \cite{Chen2016}, \cite{Brennan2018}). These problems have received widespread attention in the literature (also) because there is a region of parameters $(n,k,p,q)$ for which, while statistically possible, it appears to be computationally hard to solve the detection problem.

In this work, we focus on the case where the planted structure in $\calG(n,q)$ is a $k_{\s{R}} \times k_{\s{L}}$ bipartite subgraph with edge density $p>q$. If $p=1$, the planted structure appears as a \emph{complete} bipartite graph with $k_{\s{R}}+k_{\s{L}}$ vertices and $k_{\s{R}} k_{\s{L}}$ edges. If $p<1$, each of the $k_{\s{R}} k_{\s{L}}$ possible edges exist with probability $p$, independently of other edges. Compared to the previously mentioned complete and dense subgraph structures, it turns that the ``geometry" of bipartite graphs is richer in the following sense. Our results suggest that the planted bipartite model interpolates between two extremes: on the one hand, we have the $\s{PC}$ and $\s{PDS}$ problems, while on the other hand, we (essentially) have the planted star problem.\footnote{Recall that a star graph on $k$ vertices is a tree with one internal node and $k-1$ leaves.} Most importantly, as we explain below, the behavior of these two extremes differs when scrutinized from the statistical and computational perspectives. Specifically, in the special case where $k_{\s{R}}=\Theta(k_{\s{L}})$, that is, the planted structure is a \emph{balanced} bipartite graph, we will prove that the statistical and computational barriers are essentially the same as those of the famous $\mathsf{PC}$ and $\mathsf{PDS}$ problems. If $k_{\s{R}}\wedge k_{\s{L}}=O(\log n)$, however, then we get completely different barriers compared to the previous case; most notably, there is no hard region! This was observed also in \cite{Laurent2019}, for the special case of a star configuration and edge probabilities $p=1$ and $q = \Theta(n^{-1})$.

More generally, we conjecture that the detection boundaries for a general planted structure strongly depend on the following quantity.\footnote{In fact, some of our analysis can be carried out for a general planted structure, and it can be shown that the maximal subgraph density indeed plays a role.}
\begin{definition} [Maximal subgraph density]\label{def:maximal subgraph density}
Fix a subgraph $\Gamma = (V(\Gamma),E(\Gamma))$ with a set of vertices $V(\Gamma)$ and a set of edges $E(\Gamma)$. The maximal subgraph density of $\Gamma$ is defined as:  
\begin{align}
m{\left(\Gamma \right)} \triangleq \max_{\emptyset\neq H \subseteq \Gamma} \frac{\vert E{\left(H\right)}\vert}{\vert V{\left(H\right)}\vert}.
\end{align}
\end{definition}
Note that $\Omega \left(1\right) \leq m{\left(\Gamma \right)} \leq O \left(\vert V{\left(\Gamma\right)}\vert\right)$. The conjecture above is reasonable, as it is well-known \cite{Erdos1960,Bollobas1981} that the maximal subgraph density plays an important role in the characterization of the subgraphs count in Erd\H{o}s-R\'{e}nyi random graphs, which in turn are strongly connected to the question of impossibility and possibility of detection. Now, in continuation to the discussion above, note that for the $\mathsf{PC}$ and $\mathsf{PDS}$ problems, where $\Gamma$ is a complete graph on $k$ vertices,\footnote{In both $\mathsf{PC}$ and $\mathsf{PDS}$ we ``plant" $k$ vertices chosen uniformly at random from $[n]$; this is equivalent to planting a complete graph on $k$ vertices, and thus $\Gamma$ is indeed a clique of size $k$, irrespectively of the fact that in $\mathsf{PDS}$ we keep the edges in $\Gamma$ with probability $p$.} it is clear that the maximal subgraph density is $\frac{k-1}{2} = \Theta(k)$, namely, it is proportional to the number of vertices in the structure. In our problem, where $\Gamma$ is a complete bipartite graph, we get,
\begin{align} 
m{\left({\mathsf{K}}_{\calR, \calL}\right)} = \frac{k_{\s{R}} k_{\s{L}}}{k_{\s{R}} + k_{\s{L}}}.
\end{align}
Accordingly, in the balanced case, we get that $m{\left({\mathsf{K}}_{\calR, \calL}\right)} = \Theta \left(\vert v{\left({\mathsf{K}}_{\calR, \calL}\right)}\vert\right)$, i.e., it is proportional to the number of vertices in the structure, as in the $\mathsf{PC}$ and $\mathsf{PDS}$ problems. In the case where $k_{\s{R}}\wedge k_{\s{L}}=O(\log n)$, we get that $m{\left({\mathsf{K}}_{\calR, \calL}\right)} = O(\log n)$. 
In this paper, we analyze the bipartite subgraph detection problem, for any values of $k_{\s{R}}$ and $k_{\s{L}}$, and as so cover a spectrum of possible scaling of the maximal subgraph density; to the best of our knowledge, our work is the first to face a detection problem of a planted structure with geometry containing every possible value of the maximal subgraph density. 


\paragraph{Main contributions.} For the detection problem described above (see, Problem~\ref{prob:pdbs} for a precise definition), we characterize the statistical and computational limits of detection. To that end, we derive impossibility results in the form of information-theoretic lower bounds that applies to \emph{all} algorithms, providing conditions under which strong detection is impossible. We then propose several algorithms, analyze their performance, and show that our bounds are tight (up to poly-log factors). 
\begin{table}[t!]
\centering
\begin{tabular}{|c||c||c|}
 \hline
& $\s{Dense}\;\s{regime}$ & $\s{Sparse}\;\s{regime}$\\
 \hline
 $\s{Impossible}$  & $k_{\s{R}}\wedge k_{\s{L}} = O(\log n),k_{\s{R}}^2\vee k_{\s{L}}^2\ll n$ & $\chi^2(p||q)\ll \frac{1}{k_{\s{R}}\wedge k_{\s{L}}}\wedge\frac{n^2}{k_{\s{R}}^2k_{\s{L}}^2}\wedge\frac{n}{k^2_{\s{R}}\vee k^2_{\s{L}}}$\\
\hline
 $\s{Hard}$ &  $k_{\s{R}}\wedge k_{\s{L}} = \Omega(\log n),k_{\s{R}}^2\vee k_{\s{L}}^2\ll n$& $\frac{1}{k_{\s{R}}\wedge k_{\s{L}}}\ll\chi^2(p||q)\ll \frac{n^2}{k^2_{\s{R}} k^2_{\s{L}}}\wedge\frac{n}{k^2_{\s{R}}\vee k^2_{\s{L}}}$ \\
 \hline
 $\s{Easy}$ &  $k_{\s{R}}^2\vee k_{\s{L}}^2\gg n$& $\chi^2(p||q) \gg \frac{n^2}{k_{\s{R}}^2 k_{\s{L}}^2}\wedge\frac{n}{k^2_{\s{R}}\vee k^2_{\s{L}}}$\\
 \hline
\end{tabular}
\caption{Bipartite subgraph detection thresholds.}
\label{table:detection thresholds}
\end{table}
Similarly to other related high-dimensional problems, we observe a gap between the statistical limits we derive and the performance of the efficient algorithms we construct. We conjecture that below the computational barrier polynomial-time algorithms do not exist, namely, the detection task is computationally hard (on an average-case). To provide an evidence for this phenomenon, we follow a recent line of work \cite{hopkins2017bayesian,Hopkins18,Kunisky19,Cherapanamjeri20,gamarnik2020lowdegree,Huleihel2022}, and show that the class of low-degree polynomials fail to solve the detection problem in this conjecturally hard regime. A looser but simpler version of our bounds is given in Table~\ref{table:detection thresholds}, where we recall that $\chi^2(p||q) = \frac{(p-q)^2}{q(1-q)}$ is the $\chi^2$-divergence between two Bernoulli random variables with parameters $p$ and $q$, $a\vee b\triangleq\max\{a,b\}$ and $a\wedge b\triangleq\min\{a,b\}$, and the notation ``$\ll$" hides factors that are sub-polynomial in $n$ (see, the notation paragraph below for a precise definition). 

To illustrate our results further, let us focus here on the special case of complete bipartite, where $p=1$. From an algorithmic perspective, the detection thresholds in Table~\ref{table:detection thresholds} are achieved by a combination of the following three procedures:
\begin{itemize}
    \item \emph{Scan test}: scans for the densest bipartite subgraph in the observed graph, and compare its size to a certain threshold.
    \item \emph{Count test}: counts the number of edges in the observed graph, and compare this count to a certain threshold.
    \item \emph{Degree test}: looks for the maximal vertex degree in the observed graph, and compare this degree to a certain threshold.
\end{itemize}
Precise definitions of the above procedures are given in Section~\ref{sec:main}. It should be emphasized that the above algorithms are by no means novel; the exact same algorithms were suggested for the $\mathsf{PC}$ and $\mathsf{PDS}$ problems in, e.g., \cite{Hajek2015,Brennan2018}. 
Now, for the $\mathsf{PC}$ problem we have the phase diagram shown in Figure~\ref{fig:planted clique}. Note that the scan test is always at least as good as degree and count tests. Also, note that count and degree tests have the same performance. Finally, there is a region where solving $\mathsf{PC}$ is (conjecturally) computationally hard.
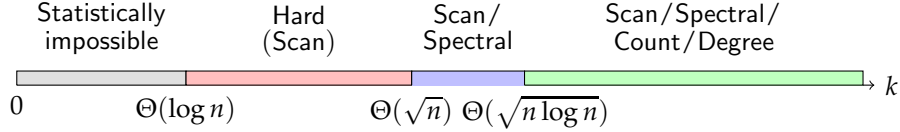
\begin{figure}[t!]
\centering

\begin{tikzpicture}[scale=1.5]
\tikzstyle{every node}=[font=\footnotesize]

\def\xUnit{5}
\def\yUnit{1.25}
\def\xmin{0}
\def\xmax{1.5 * \xUnit + 0.1}
\def\ymin{0}

\draw[->] (\xmin,\ymin) -- (\xmax,\ymin) node[right] {$k$};

\node at (\xUnit * 0.7, 0) [below] {$\Theta(\sqrt{n})$};
\node at (\xUnit * 0.92, 0) [below] {$\Theta(\sqrt{n\log n})$};
\node at (\xUnit * 0.3, 0) [below] {$\Theta(\log{n})$};
\node at (0, 0) [below] {$0$};

\filldraw[fill=gray!25, draw=black] (0, 0) -- (0, 0.1 * \yUnit) -- (0.3 * \xUnit,  0.1 * \yUnit) -- (0.3 * \xUnit, 0) -- (0, 0);
\filldraw[fill=red!25, draw=black] (0.3 * \xUnit, 0) -- (0.3 * \xUnit, 0.1 * \yUnit) -- (0.7 * \xUnit, 0.1 * \yUnit) -- (0.7 * \xUnit, 0) -- (0.3 * \xUnit, 0);
\filldraw[fill=blue!25, draw=black] (0.7 * \xUnit, 0) -- (0.7 * \xUnit, 0.1 * \yUnit) -- (0.9 * \xUnit, 0.1 * \yUnit) -- (0.9 * \xUnit, 0) -- (0.9 * \xUnit, 0);
\filldraw[fill=green!25, draw=black] (0.9 * \xUnit, 0) -- (0.9 * \xUnit, 0.1 * \yUnit) -- (1.5 * \xUnit, 0.1 * \yUnit) -- (1.5 * \xUnit, 0) -- (0.9 * \xUnit, 0);

\node at (\xUnit * 0.15 , \yUnit * 0.5) {$\s{{Statistically}}$};
\node at (\xUnit * 0.15 , \yUnit * 0.3) {$\s{{impossible}}$};
\node at (\xUnit * 0.5 , \yUnit * 0.5) {$\s{{Hard}}$};
\node at (\xUnit * 0.5 , \yUnit * 0.3) {$(\s{{Scan}})$};
\node at (\xUnit * 0.8 , \yUnit * 0.5) {$\s{{Scan/}}$};
\node at (\xUnit * 0.8 , \yUnit * 0.3) {$\s{{Spectral}}$};

\node at (\xUnit * 1.2 , \yUnit * 0.5) {$\s{{Scan/Spectral/}}$};
\node at (\xUnit * 1.2 , \yUnit * 0.3) {$\s{{Count/Degree}}$};

\end{tikzpicture}

\caption{Phase diagram for detecting a planted clique of size $k$, as a function of $k$.}
\label{fig:planted clique}
\end{figure}
For the complete bipartite case, we obtain the phase diagram shown in Figure~\ref{fig:planted bipartite}, in which the axes correspond to the number of vertices in the right and left planted sets. As compared to the $\s{PC}$ case, it can be seen that the obtained phase diagram is much more involved. For example, note that there exists a region of parameters, i.e., $k_{\s{R}}\wedge k_{\s{R}} = O(\log n)$, where the detection problem is either statistically impossible or possible to solve in polynomial-time using the degree test. This was observed also in \cite{Laurent2019}, for the special case of a star configuration, where $k_{\s{R}}\wedge k_{\s{R}} = 1$, and edge probabilities $p=1$ and $q=\Theta(n^{-1})$.

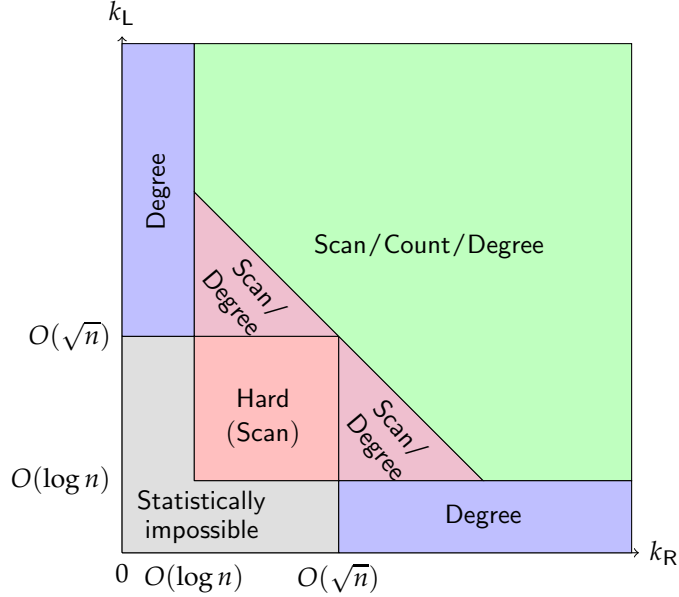
\begin{figure}[t!]
\centering

\begin{tikzpicture}[scale=1.5]
\tikzstyle{every node}=[font=\footnotesize]
\def\xUnit{3.2}
\def\yUnit{0.85}
\def\xmin{0}
\def\ymin{0}
\def\xmax{1.4 * \xUnit + 0.1}
\def\ymax{1.4 * \xUnit + 0.1}

\draw[->] (\xmin,\ymin) -- (\xmax,\ymin) node[right] {$k_{\s{R}}$};
\draw[->] (\xmin,\ymin) -- (\xmin,\ymax) node[above] {$k_{\s{L}}$};

\node at (\xUnit * 0.6, 0) [below] {$O(\sqrt{n})$};
\node at (\xUnit * 0.2, 0) [below] {$O(\log{n})$};
\node at (0, 0) [below] {$0$};

\node at (0, \xUnit * 0.6) [left] {$O(\sqrt{n})$};
\node at (0, \xUnit * 0.2) [left] {$O(\log{n})$};

\filldraw[fill=gray!25, draw=black] (0, 0) -- (0, 0.6 * \xUnit) -- (0.2 * \xUnit,  0.6 * \xUnit) -- (0.2 * \xUnit,  0.2 * \xUnit) -- (0.6 * \xUnit,  0.2 * \xUnit) -- (0.6 * \xUnit, 0) -- (0, 0);
\filldraw[fill=blue!25, draw=black] (0, 0.6 * \xUnit) -- (0,  \xmax-\xUnit * 0.02) -- (0.2 * \xUnit,  \xmax-\xUnit * 0.02) -- (0.2 * \xUnit,  0.6 * \xUnit) -- (0, 0.6 * \xUnit);
\filldraw[fill=blue!25, draw=black] (0.6 * \xUnit, 0) -- (\xmax-\xUnit * 0.02, 0) -- (\xmax-\xUnit * 0.02, 0.2 * \xUnit) -- (0.6 * \xUnit, 0.2 * \xUnit) -- (0.6 * \xUnit, 0);
\filldraw[fill=red!25, draw=black] (0.2 * \xUnit, 0.2 * \xUnit) -- (0.2 * \xUnit, 0.6 * \xUnit) -- (0.6 * \xUnit, 0.6 * \xUnit) -- (0.6 * \xUnit, 0.2 * \xUnit) -- (0.2 * \xUnit, 0.2 * \xUnit);
\filldraw[fill=purple!25, draw=black] (0.2 * \xUnit, 0.6 * \xUnit) -- (0.2 * \xUnit, 1 * \xUnit) -- (0.6 * \xUnit, 0.6 * \xUnit) -- (0.2 * \xUnit, 0.6 * \xUnit);
\filldraw[fill=purple!25, draw=black] (0.6 * \xUnit, 0.2 * \xUnit) -- (1 * \xUnit, 0.2 * \xUnit) -- (0.6 * \xUnit, 0.6 * \xUnit) -- (0.6 * \xUnit, 0.2 * \xUnit);
\filldraw[fill=green!25, draw=black] (0.2 * \xUnit, 1 * \xUnit) -- (0.2 * \xUnit, \xmax-\xUnit * 0.02) -- (\xmax-\xUnit * 0.02, \xmax-\xUnit * 0.02) -- (\xmax-\xUnit * 0.02, 0.2 * \xUnit) -- (\xUnit, 0.2 * \xUnit) -- (0.2 * \xUnit, 1 * \xUnit);

\node at (\xUnit * 0.22 , \xUnit * 0.14) {$\s{{Statistically}}$};
\node at (\xUnit * 0.22 , \xUnit * 0.06) {$\s{{impossible}}$};
\node at (\xUnit * 0.39 , \xUnit * 0.43) {$\s{{Hard}}$};
\node at (\xUnit * 0.39 , \xUnit * 0.33) {$(\s{{Scan}})$};
\node[rotate=-50] at (\xUnit * 0.38 , \xUnit * 0.73) {$\s{{Scan/}}$};
\node[rotate=-50] at (\xUnit * 0.3 , \xUnit * 0.69) {$\s{Degree}$};
\node[rotate=-50] at (\xUnit * 0.77 , \xUnit * 0.34) {$\s{{Scan/}}$};
\node[rotate=-50] at (\xUnit * 0.7 , \xUnit * 0.29) {$\s{Degree}$};
\node at (\xUnit * 1, \xUnit * 0.1) {$\s{Degree}$};
\node[rotate=90] at (\xUnit * 0.1, \xUnit * 1) {$\s{Degree}$};
\node at (\xUnit * 0.85, \xUnit * 0.85) {$\s{{Scan/Count/Degree}}$};

\end{tikzpicture}

\caption{Phase diagram for detecting a planted $k_{\s{R}} \times k_{\s{L}}$ complete bipartite.}
\label{fig:planted bipartite}
\end{figure}


\paragraph{Related Work.} The $\mathsf{PC}$ detection problem has been studied from many different perspectives in the last decades. Many polynomial-time algorithms have been proposed, including simple edge count, degree test \cite{Kucera1995}, that achieves detection if $k=\Omega\left(\sqrt{n \log{n}}\right)$, and more sophisticated spectral approaches, e.g., \cite{Alon1998} that succeed if $k=\Omega\left(\sqrt{n}\right)$. Statistical limits as well as the analysis of inferring algorithms are also studied in other subgraphs detection problems, including, e.g., community detection \cite{Castro2014}, \cite{Verzelen2015}. It was first conjectured in \cite{Jerrum1992} that the $\mathsf{PC}$ problem cannot be solved in polynomial-time if $k=o\left(\sqrt{n}\right)$, while an exhaustive scan succeeds if $k=\Omega\left(\log{n}\right)$. This $\mathsf{PC}$  conjecture has been used through reduction arguments in other high-dimensional problems, such as sparse $\mathsf{PCA}$ \cite{Berthet2013} and $\mathsf{PDS}$ \cite{Hajek2015}, who observed an impossible-hard-easy phase transition phenomena. For recent surveys we refer the reader to \cite{Abbe,Wu2018StatisticalPW}, and many references therein. 

Most of the above high-dimensional detection problems typically deal with a planted structure whose elements are \emph{all} either fully connected ($\mathsf{PC}$) or connected with some probability (e.g., $\mathsf{PDS}$). Accordingly, existing techniques used to determine the statistical and computational limits are specialized to these types of structures. There are, however, some exceptions. For example, in \cite{massoulie19a} the detection and recovery task of $\s{D}$-ary trees superimposed in a \emph{sparse} Erd\H{o}s-R\'{e}nyi random graph, where $p=1$ and $q=\Theta(n^{-1})$, were considered. As mentioned above, a spacial case of \cite{massoulie19a} is the star subgraph, which is captured also by our model. Another exception is \cite{Bagaria20}, where the problem of planted Hamiltonian cycle recovery was addressed. Recently in \cite{Huleihel2022} the detection and recovery problems of a \emph{general} planted \emph{induced} subgraph was considered; however, since the latent structure appears as an induced subgraph, the techniques and results in \cite{Huleihel2022} are irrelevant to our setting, where we take a union of bipartite graph and the underlying Erd\H{o}s-R\'enyi random graph. While our bipartite framework is still specific, it turns out, as mentioned above, that it covers asymptotically the entire range of possible maximal subgraph densities, and that its analysis is technically interesting. Moreover, bipartite graphs have many theoretical and practical applications (see, e.g., \cite{Asratian1998}, \cite{Pavlopoulos2018}), which provide further motivation for this paper.

Over the last few years, there has been a success in developing a rigorous notion of what can and cannot be achieved by efficient algorithms. Recent results, e.g., \cite{berthet2013complexity,ma2015computational,cai2015computational,krauthgamer2015semidefinite,hajek2015computational,chen2016statistical,wang2016average,wang2016statistical,gao2017sparse,brennan18a,brennan19,wu2018statistical,brennan20a,hopkins2017bayesian,Hopkins18,Kunisky19,Cherapanamjeri20,gamarnik2020lowdegree,barak2016nearly,deshpande2015improved,meka2015sum,TengyuWig15,kothari2017sum,hopkins2016integrality,raghavendra2019highdimensional,hopkins2017power,mohanty2019lifting,Feldman17,feldman2018complexity,Diakonikolas17,DiakonikolasKong19,Lenka16,Lesieur_2015,Lesieur_2016,Krzakala10318,Ricci_Tersenghi_2019,bandeira2018notes} revealed an intriguing phenomenon that is common to many high-dimensional problems with a planted structure: there is an inherent gap between the amount of data needed by all computationally efficient algorithms and what is needed for statistically optimal algorithms. Various forms of rigorous evidence for this phenomenon, i.e., hardness in the average-case sense, have been proposed, and they can roughly be classified into two groups: 1) \emph{Failure under certain computation models}, namely, showing that powerful classes of computationally efficient algorithms, such as, low-degree polynomials \cite{hopkins2017bayesian,Hopkins18,Kunisky19,Cherapanamjeri20,gamarnik2020lowdegree}, sum-of-squares hierarchy \cite{barak2016nearly,deshpande2015improved,meka2015sum,TengyuWig15,kothari2017sum,hopkins2016integrality,raghavendra2019highdimensional,hopkins2017power,mohanty2019lifting}, statistical query algorithms \cite{Feldman17,feldman2018complexity,Diakonikolas17,DiakonikolasKong19}, message-passing algorithms \cite{Lenka16,Lesieur_2015,Lesieur_2016,Krzakala10318,Ricci_Tersenghi_2019,bandeira2018notes}, etc., fail in the conjectured computationally hard regime of the problem. 2) \emph{Average-case reductions} from another problem, such as the planted clique problem, conjectured to be computationally hard, e.g., \cite{berthet2013complexity,ma2015computational,cai2015computational,chen2016statistical,hajek2015computational,wang2016average,wang2016statistical,gao2017sparse,brennan18a,brennan19,wu2018statistical,brennan20a}. As mentioned above, in this paper we use the machinery of low-degree polynomials to provide an evidence for the statistical computational gap conjectured in our model. 


\paragraph{Notation.} For a simple undirected graph $\s{G}$, denote its adjacency matrix by $\mathsf{A}$. Also, let $V\left(\s{G}\right)$ and $E\left(\s{G}\right)$ be the vertex set and the edges set of $\s{G}$, respectively. We denote the size of any finite set $\calS$ by $\vert \calS \vert$. Let $m{\left(\s{G}\right)}$ be the \emph{maximal subgraph density} of $\s{G}$, as defined in Definition~\ref{def:maximal subgraph density}. Moreover, let $\Delta{\left(\s{G}\right)}$ be the \emph{maximal vertex degree} of $\s{G}$, that is, 
$\Delta{\left(\s{G}\right)} \triangleq \max_i \mathsf{deg}_{i}$, where $\mathsf{deg}_{i}$ is the degree of the $i^{\s{th}}$ vertex, that is, the number of edges with one end point at vertex $i$. 
Let $\mathsf{N}_{{\mathsf{K}}_{\calR, \calL}}\triangleq\binom{n}{k_{\s{R}}}\binom{n-k_{\s{R}}}{k_{\s{L}}}$ denote the number of all possible $k_{\s{R}} \times k_{\s{L}}$ copies of bipartite graphs in the complete graph over $[n]$ vertices; we label those copies as ${\mathsf{K}}^{\left(1\right)}_{\calR, \calL}, {\mathsf{K}}^{\left(2\right)}_{\calR, \calL},\ldots,{\mathsf{K}}_{\calR, \calL}^{(\mathsf{N}_{{\mathsf{K}}_{\calR, \calL}})}$, and denote by $\calK_{n,k_{\s{R}
},k_{\s{L}}}$ the set of all these copies. Let $\left[n\right]$ denote the set of integers $\{1,2,\ldots,n\}$, for any $n\in\mathbb{N}$. For probability measures $\mathbb{P}$ and $\mathbb{Q}$, let $d_{\s{TV}}(\mathbb{P},\mathbb{Q})=\frac{1}{2}\int |\mathrm{d}\mathbb{P}-\mathrm{d}\mathbb{Q}|$ and $\chi^2(\mathbb{P}||\mathbb{Q}) = \int\frac{(\mathrm{d}\mathbb{P}-\mathrm{d}\mathbb{Q})^2}{\mathrm{d}\mathbb{Q}}$ denote the total variation distance and the $\chi^2$-divergence, respectively. Let $\s{Bern}(p)$ and $\s{Binomial}(n,p)$ denote the Bernoulli and Binomial distributions with parameter $p$ and $n$, respectively. We use
standard asymptotic notation: for two positive sequences $\{a_n\}$ and $\{b_n\}$, we write $a_n = O(b_n)$ if $a_n\leq Cb_n$, for some absolute constant $C$ and for all $n$; $a_n = \Omega(b_n)$, if
$b_n = O(a_n)$; $a_n = \Theta(b_n)$, if $a_n = O(b_n)$ and $a_n = \Omega(b_n)$, $a_n = o(b_n)$ or $b_n = \omega(a_n)$, if $a_n/b_n\to0$, as $n\to\infty$. The notation $\ll$ refers to polynomially less than in $n$, namely, $a_n\ll b_n$  if $\liminf_{n\to\infty}\log_n a_n<\liminf_{n\to\infty}\log_n b_n$, e.g., $n\ll n^2$, but $n\not\ll n\log_2 n$. Finally, for $a,b\in\mathbb{R}$, we let $a\vee b\triangleq\max\{a,b\}$ and $a\wedge b\triangleq\min\{a,b\}$.

\section{Problem Formulation} \label{sec:model}
Let us describe the general setting of planted bipartite graph detection. We have a total of $n$ vertices. We pick $k_{\s{R}}$ out of $n$ vertices, uniformly at random, and denote this set by $\calR \subset \left[n\right]$. From the remaining $n-k_{\s{R}}$ vertices, we again pick another set of $k_{\s{L}}$ vertices, uniformly at random, and denote the obtained set by $\calL \subseteq \left[n\right] \backslash \calR$. We refer to $\calR$ and $\calL$ as the \emph{right} and \emph{left} planted sets, respectively. Let $\mathscr{E}\left(\calR, \calL\right)$ be the set of edges $\left \{\left(i,j\right) : i<j, i \in \calR, j \in \calL \text{ or } i \in \calL, j \in \calR\right\}$, and let ${\mathsf{K}}_{\calR, \calL}$ be an empty graph on $n$ vertices, except for all the elements in $\mathscr{E}\left(\calR, \calL\right)$. Let $V\left({\mathsf{K}}_{\calR, \calL}\right) \triangleq \calR \cup \calL$ and $E\left({\mathsf{K}}_{\calR, \calL}\right) \triangleq \mathscr{E}\left(\calR, \calL\right)$ denote the set of vertices and edges of $\mathsf{K}_{\calR, \calL}$, respectively. Note that $\abs{V\left({\mathsf{K}}_{\calR, \calL}\right)} = k_{\s{R}}+k_{\s{L}}$ and $\abs{E\left({\mathsf{K}}_{\calR, \calL}\right) } = k_{\s{R}} k_{\s{L}}$. We refer to ${\mathsf{K}}_{\calR, \calL}$ as the planted structure. Finally, let ${\mathsf{K}}^{p}_{\calR, \calL}$ be the random graph obtained by keeping each edge in $E\left({\mathsf{K}}_{\calR, \calL}\right)$, with probability $p$. Consequently, our detection problem can be phrased as the following simple hypothesis testing problem: under the null hypothesis $\calH_0$, the observed graph $\s{G}$ is an Erd\H{o}s-R\'{e}nyi random graph $\calG(n,q)$ with edge density $0<q<1$, which might be a function of $n$. Under the alternative hypothesis $\calH_1$, we first draw a random bipartite graph ${\mathsf{K}}^{p}_{\calR, \calL}$ on $n$ vertices, as described above, and then join every other two vertices, such that at least one of the vertices is outside $\calL\cup\calR$, with probability $q$; the resulted graph is $\s{G}$. Thus, the probability of a planted edge is $p$, and the probability of a non-planted edge is $q$. Without loss of generality, we assume throughout that $p>q$.\footnote{Any achievability or converse result for the $p > q$ case imply a corresponding result for $p<q$ by flipping each edge in the observed graph.} Another equivalent way to define the alternative distribution is as follows: a base graph $\s{\tilde{G}}$ is sampled from $\calG(n,q)$, and we observe the union of this base graph with ${\mathsf{K}}^{p'}_{\calR,\calL}$, i.e., $\s{G}=\s{\tilde{G}}\cup {\mathsf{K}}^{p'}_{\calR, \calL}$, where $p' = \frac{p-q}{1-q}$.\footnote{The probability $p'$ is chosen such that the probability of a planted edge in the union graph is $p$, i.e., $p = p'+q(1-p')$.} We denote the ensemble of graphs under the alternative hypothesis as $\calG(n,k_{\s{R}},k_{\s{L}},p,q)$. 
\begin{problem}[Detection problem]\label{prob:pdbs}
The planted dense bipartite subgraph (PDBS) detection problem, denoted by $\mathsf{PDBS}{\left(n,k_{\s{R}},k_{\s{L}},p,q\right)}$, refers to the problem of distinguishing between the hypotheses:
\begin{align}
\calH_0: \s{G} \sim \calG(n,q) \quad \mathsf{vs}. \quad \calH_1 : \s{G} \sim \calG(n,k_{\s{R}},k_{\s{L}},p,q).\label{eqn:super_hypo}   
\end{align}     
\end{problem}
We study the above framework in the asymptotic regime where $n\to \infty$, and $(k_{\s{R}},k_{\s{L}},p,q)$ may also change as a function of $n$. Throughout the paper we assume that $k_{\s{R}}+k_{\s{L}} = o\left(n\right)$. Also, ``dense regime" and ``sparse regime" correspond to the cases where $p,q=\Theta(1)$ and $p,q = \Theta(n^{-\alpha})$, for $\alpha\in(0,2]$, respectively. Observing $\s{G}$, the goal is to design a test $\phi:\ppp{0,1}^{\binom{n}{2}}\to \{0, 1\}$ that distinguishes between $\calH_0$ and $\calH_1$. Specifically, the average $\mathsf{Type}$ $\mathsf{I}$+$\mathsf{II}$ risk of a test $\phi$ is defined as
$\mathsf{R}_n (\phi) = \pr_{\calH_0}(\phi(\s{G}) = 1)+ \pr_{\calH_1}(\phi(\s{G}) = 0)$. A sequence of tests $\phi_n$ indexed by $n$ succeeds if $\lim_{n \to \infty} \mathsf{R}_n(\phi_n) = 0$, while if $\lim_{n \to \infty} \mathsf{R}_n(\phi_n)>0$, it fails. Finally, we denote the optimal risk by $\mathsf{R}_n^\star$, i.e., $\mathsf{R}_n^\star\triangleq\inf_{\phi_n}\mathsf{R}_n(\phi_n)$.  The above is summarized in the following definition.
\begin{definition}[Strong detection]\label{def:detection}
Let $\pr_{\calH_0}$ and $ \pr_{\calH_1}$ be the distributions of $\s{G}$ under the uniform and planted hypotheses, respectively. A possibly randomized algorithm $\phi_n(\s{G}) \in \{0, 1\}$ achieves strong detection, if its $\mathsf{Type}$ $\mathsf{I}$+$\mathsf{II}$ risk satisfies $\lim_{n \to \infty} \mathsf{R}_n(\phi_n)=0$. Conversely, if for all sequences of possibly randomized algorithms $\phi_n(\s{G}) \in \{0, 1\}$, the $\mathsf{Type}$ $\mathsf{I}$+$\mathsf{II}$ risk satisfies $\liminf_{n \to \infty} \mathsf{R}_n(\phi_n)>0$, then strong detection is impossible.
\end{definition}

\section{Main Results}\label{sec:main}
In this section, we determine the statistical and computational limits of the planted bipartite subgraph detection problem. We start by presenting our information-theoretic lower bounds, followed by our algorithmic upper bounds. Then, we state our computational lower bounds, and finally, present several phase diagrams in order to illustrate our results as a function of the various parameters in our model. 

\paragraph{Statistical lower bounds.} The following result give conditions under which strong detection is statistically (or, information-theoretically) impossible. For simplicity of notation, we define
\begin{align}
    \gamma_n(x,y)\triangleq\log\p{1+\frac{n}{xy}\frac{\log 2}{2}},
\end{align}
for any $x,y\in\mathbb{N}$. Also, recall that the $\chi^2$-divergence between two Bernoulli random variables $\s{Bern}(p)$ and $\s{Bern}(q)$ is given by $\chi^2(p||q) = \frac{(p-q)^2}{q(1-q)}$. 
\begin{theorem}[Statistical lower bounds] \label{thm:sparse}  
Consider the $\mathsf{PDBS}{\left(n,k_{\s{R}},k_{\s{L}},p,q\right)}$ detection problem in Problem~\ref{prob:pdbs}. 
Then, strong detection is impossible with $\s{R}_n^\star>1/2$, if:
\begin{subequations}
  \begin{align}
  &\chi^2(p||q)\leq \frac{n\cdot\gamma_n(k_{\s{R}},k_{\s{L}})}{8k_{\s{R}} k_{\s{L}}},\label{eqn:sparsecond1}\\
  \intertext{and at least one of the following two conditions hold:}
  &\chi^2(p||q)\leq \frac{\gamma_n(k_{\s{R}},k_{\s{R}})}{2k_{\s{L}}}\vee\frac{\gamma_n(k_{\s{L}},k_{\s{L}})}{2k_{\s{R}}},\label{eqn:densecond1}\\
   &\chi^2(p||q)\leq\p{\frac{1}{2k_{\s{L}}}\wedge\frac{n\cdot \gamma_n(k_{\s{L}},k_{\s{L}})}{8k_{\s{R}}^2}}\vee\p{\frac{1}{2k_{\s{R}}}\wedge\frac{n\cdot\gamma_n(k_{\s{R}},k_{\s{R}})}{8k_{\s{L}}^2}}.\label{eqn:sparsecond2}
  \end{align}\label{eqn:sparsecond}
\end{subequations}
\end{theorem}

A few remarks are in order. First, it should be emphasized that while the theorem above hold for any values of $p$ and $q$, we will see later on that when specified to the dense and sparse regimes, some of the conditions in \eqref{eqn:sparsecond} become redundant. Second, we will see that in the polynomial-scale, the thresholds in \eqref{eqn:sparsecond} become much simpler analytically, and easy to interpret.  
Finally, we would like to mention here that the various constants in Theorem~\ref{thm:sparse} are not optimal; they were obtained with the goal of making the proof transparent, and because we focus mainly on the asymptotic behaviour of the thresholds for detection. 
\paragraph{Upper bounds.} We now describe our upper bounds. To that end, we propose three algorithms and analyze their performance. Define the statistics:
\begin{align}
\mathsf{Scan}{\left(\s{A}\right)} &\triangleq \max_{(\calR',\calL')\in\calS_{k_{\s{R}},k_{\s{L}}}}{\sum_{\substack{i \in \calR', j \in \calL'}}\s{A}_{ij}},\label{algo:scan non-complete}\\
\mathsf{Count}{\left(\s{A}\right)}&\triangleq \sum_{i<j} \s{A}_{ij},\label{algo:count}\\
    \mathsf{MaxDeg}{\left(\s{A}\right)}&\triangleq \max_{1\leq i\leq n}\sum_{j=1}^n \s{A}_{ij},\label{algo:degree}
\end{align}
where $\calS_{k_{\s{R}},k_{\s{L}}}\triangleq\{(\calR',\calL')\subset\left[n\right]\times\left[n\right]:\vert \calR' \vert = k_{\s{R}},\vert \calL' \vert = k_{\s{L}}, \calR'\cap\calL'=\emptyset\}$. Then, our proposed tests are $\phi_{\s{Scan}}\triangleq\Ind\ppp{\mathsf{Scan}{\left(\s{A}\right)}\geq \tau_{\mathsf{Scan}}}$, $\phi_{\s{Count}}\triangleq\Ind\ppp{\mathsf{Count}{\left(\s{A}\right)}\geq \tau_{\mathsf{Count}}}$, and $\phi_{\s{Deg}}\triangleq\Ind\ppp{\mathsf{MaxDeg}{\left(\s{A}\right)}\geq \tau_{\mathsf{Deg}}}$. As mentioned in the introduction, the tests above are rather classic and were proposed in many other related problems, e.g., \cite{kolar2011minimax,butucea2013detection,ma2015computational,Castro2014,Brennan2018,Hajek2015,Huleihel2022,pmlr-v99-brennan19a}. Also, note that the count and maximum degree tests exhibit polynomial computational complexity of $O(n^2)$ operations, and hence efficient. This is not the case for the scan test, which exhibits an exponential computational complexity, and thus inefficient. Indeed, the search space in \eqref{algo:scan non-complete} is of cardinality $|\calS_{k_{\s{R}},k_{\s{L}}}|=\binom{n}{k_{\s{R}}}\binom{n-k_{\s{R}}}{k_{\s{L}}}$, which is at least quasi-polynomial already when $k_{\s{R}}\vee k_{\s{L}} = \omega(1)$. The following result provides sufficient conditions under which the risk of each of the tests above is small. 
\begin{theorem}[Algorithmic upper bounds] \label{thm:pdbs}
Fix $\delta \in \left(0,1\right).$ Consider the $\mathsf{PDBS}{\left(n,k_{\s{R}},k_{\s{L}},p,q\right)}$ detection problem in Problem~\ref{prob:pdbs}, and suppose that $p$ and $q$ are such that $\vert p-q \vert = O{\left(q\right)}$ and $q<c<1$, for some constant $0<c<1$. Then,
\begin{enumerate}
    \item For $\tau_{\mathsf{Scan}} = k_{\s{R}} k_{\s{L}} \frac{p+q}{2}$, we have $\s{R}_n(\phi_{\s{Scan}}) \leq \delta$, provided that,
\begin{align}
\chi^2(p||q) = \Omega {\left(\frac{\log{\left(\frac{n}{k_{\s{R}}\vee k_{\s{L}}}\right)}+\frac{\log{\frac{2} {\delta}}}{k_{\s{R}}\vee k_{\s{L}}}}{k_{\s{R}}\wedge k_{\s{L}}}\right)}.\label{eqn:scanCondmain}
\end{align}
\item For $\tau_{\mathsf{Count}}=\binom{n}{2}q + k_{\s{R}} k_{\s{L}} \frac{p-q}{2}$, we have $\s{R}_n(\phi_{\s{Count}}) \leq \delta$, provided that,
\begin{align}
\chi^2(p||q) =\Omega{\left(\frac{n^2}{k^2_{\s{R}} k^2_{\s{L}}} \cdot \log{\frac{2}{\delta}}\right)}.\label{eqn:countCondmain}
\end{align}
\item For $\tau_{\mathsf{Deg}}=\left(n-1\right)q + k_{\s{R}}\vee k_{\s{L}} \frac{p-q}{2}$, we have $\s{R}_n(\phi_{\s{Deg}}) \leq \delta$, provided that,
\begin{align}
\chi^2(p||q) = \Omega{\left(\frac{n}{k^2_{\s{R}}\vee k^2_{\s{L}}} \cdot \left(\log n + \log{\frac{2}{\delta}}\right)\right)}.\label{eqn:degCondmain}
\end{align}
\end{enumerate}
\end{theorem}

At this point, we would like to reduce the bounds in Theorems~\ref{thm:sparse} and \ref{thm:pdbs} to the dense and sparse regimes, starting with the former. In the dense case, we have $p,q=\Theta(1)$, and therefore $\chi^2(p||q) = \Theta(1)$ as well. Due to the fact that $k_{\s{R}},k_{\s{L}} = o(n)$, it is clear that the condition in \eqref{eqn:sparsecond1} is looser than the condition in \eqref{eqn:densecond1}, and thus, among the two only the later prevail. Accordingly, strong detection is impossible if \eqref{eqn:densecond1} or \eqref{eqn:sparsecond2} hold; as we show next \eqref{eqn:densecond1} can be complemented by nearly matching upper bounds, and thus \eqref{eqn:sparsecond2} is not needed. This is most easily seen when some of the poly-log factors are neglected. Specifically, note that when $k_{\s{R}}^2\vee k_{\s{L}}^2\gg n$ then \eqref{eqn:densecond1} can never hold because $\chi^2(p||q) = \Theta(1)$. Therefore, in this case it must be that $k_{\s{R}}^2\vee k_{\s{L}}^2\ll n$, in which case, \eqref{eqn:densecond1} boils down to the condition that $k_{\s{R}}\wedge k_{\s{L}} = O(\log n)$. To conclude, in the dense regime, strong detection is statistically impossible if,
\begin{align}
   k_{\s{R}}\wedge k_{\s{L}} = O(\log n)\quad\s{and}\quad k_{\s{R}}^2\vee k_{\s{L}}^2\ll n.\label{eqn:denselower}
\end{align}
This mimics the asymptotic performance of the scan and degree tests. Indeed, using Theorem~\ref{thm:pdbs}, we see from \eqref{eqn:scanCondmain} and \eqref{eqn:degCondmain} that strong detection is possible if
\begin{align}
   k_{\s{R}}\wedge k_{\s{L}} = \Omega(\log n)\quad\s{or}\quad k_{\s{R}}^2\vee k_{\s{L}}^2\gg n.\label{eqn:denseupper}
\end{align}
Specifically, the condition in the left-hand-side of \eqref{eqn:denseupper} guarantees that  $\s{R}_n(\phi_{\s{Scan}})\to0$, as $n\to\infty$, while the condition in the right-hand-side of \eqref{eqn:denseupper} guarantees that $\s{R}_n(\phi_{\s{Deg}})\to0$, as $n\to\infty$. 
Note that the above results imply that in the dense regime, the statistical barrier depends on the bipartite planted graph through the quantities $k_{\s{R}}\wedge k_{\s{L}}$ and $k_{\s{R}}\vee k_{\s{L}}$, which essentially are, up to a constant factor, the maximal subgraph density and the maximal vertex degree of the planted structure, respectively. Also, we would like to mention here that while in the above discussion we focused on the polynomial scale, especially, with regard to the conditions at the right-hand-side of \eqref{eqn:denselower} and \eqref{eqn:denseupper}, Theorems~\ref{thm:sparse} and \ref{thm:pdbs} give logarithmic factors as well. Finally, it is interesting to note that in the dense regime, the count test is redundant. This is not the case in the sparse regime as we explain next.

We move forward to the sparse regime, where $p,q = \Theta(n^{-\alpha})$, for some $\alpha\in(0,2]$, implying that $\chi^2(p||q) = \Theta(n^{-\alpha})$. In this case, we need both conditions \eqref{eqn:sparsecond1} and \eqref{eqn:sparsecond2} in Theorem~\ref{thm:sparse}. Let us simplify those conditions in the polynomial scale, as we for dense case. To that end, we use the fact that $\gamma_n(x,y)=\Omega(1\wedge\frac{n}{xy})$. Then, \eqref{eqn:sparsecond1} can be written as $\chi^2(p||q)= O\p{\frac{n}{k_{\s{R}}k_{\s{L}}}\wedge \frac{n^2}{k_{\s{R}}^2k_{\s{L}}^2}}$. Similarly, \eqref{eqn:sparsecond2} can be written as the union between $\chi^2(p||q)= O\p{\frac{1}{k_{\s{L}}}\wedge\frac{n}{k^2_{\s{R}}}\wedge \frac{n^2}{k_{\s{R}}^2k_{\s{L}}^2}}$ and $\chi^2(p||q)= O\p{\frac{1}{k_{\s{R}}}\wedge\frac{n}{k^2_{\s{L}}}\wedge \frac{n^2}{k_{\s{R}}^2k_{\s{L}}^2}}$. Intersecting these conditions we get that strong detection is impossible if $\chi^2(p||q)= O\p{\frac{1}{k_{\s{L}}}\wedge\frac{n^2}{k_{\s{R}}^2k_{\s{L}}^2}\wedge\frac{n}{k^2_{\s{R}}}}$ or $\chi^2(p||q)= O\p{\frac{1}{k_{\s{R}}}\wedge\frac{n^2}{k_{\s{R}}^2k_{\s{L}}^2}\wedge\frac{n}{k^2_{\s{L}}}}$, where we have used the fact that asymptotically $\frac{n}{k_{\s{R}}k_{\s{L}}}$ can never survive the minimum as $\frac{1}{k_{\s{R}}},\frac{1}{k_{\s{L}}}=o\p{\frac{n}{k_{\s{R}}k_{\s{L}}}}$, due to the assumption that $k_{\s{R}},k_{\s{L}}=o(n)$. Finally, it is not difficult to check that the union of the last two conditions give,
\begin{align}
    \chi^2(p||q)= O\p{\frac{1}{k_{\s{R}}\wedge k_{\s{L}}}\wedge\frac{n^2}{k_{\s{R}}^2k_{\s{L}}^2}\wedge\frac{n}{k^2_{\s{R}}\vee k^2_{\s{L}}}}.\label{eqn:sparselower}
\end{align}
This barrier mimics the asymptotic performance of the scan, count, and degree tests. Indeed, using Theorem~\ref{thm:pdbs}, we see from \eqref{eqn:scanCondmain}, \eqref{eqn:countCondmain}, and \eqref{eqn:degCondmain} that strong detection is possible if
\begin{align}
   \chi^2(p||q) = \Omega\left(\frac{\log n}{k_{\s{R}}\wedge k_{\s{L}}}\wedge\frac{n^2\log n}{k_{\s{R}}^2 k_{\s{L}}^2}\wedge\frac{n\log n}{k^2_{\s{R}}\vee k^2_{\s{L}}}\right).\label{eqn:sparseupper}
\end{align}
Indeed, each of the terms from left to right inside the $\Omega(\cdot)$ expression at the right-hand-side of \eqref{eqn:sparseupper}, guarantee that $\s{R}_n(\phi_{\s{Scan}})\to0$, $\s{R}_n(\phi_{\s{Count}})\to0$, and $\s{R}_n(\phi_{\s{Deg}})\to0$, respectively, as $n\to\infty$. Later on, we illustrate the above bounds using phase diagrams that illustrate the tradeoff between the various parameters as a function of $n$. One evident and important observation here is that both the efficient count and degree tests, and the exhaustive scan test, are needed to attain the information-theoretic lower bounds (up to poly-log factors). As discussed above, however, the scan test is not efficient. Accordingly, from the computational point of view, one may wonder whether the inefficient scan test can be replaced by an efficient test with the same performance guarantees, i.e., one that achieves the same statistical barrier as that of the scan test. We next give evidence that, based on the \emph{low-degree polynomial conjecture}, efficient algorithms that run in polynomial-time do not exist in the regime where the scan test succeeds and the count and degree tests fail. 

\paragraph{Computational lower bounds.} Inspecting our lower and upper bounds above, it can be seen that there is a gap in terms of what can be achieved statistically and efficiently. In particular, in the region where $k^2_{\s{R}}\vee k^2_{\s{L}} \ll n$, for the dense regime, or the region where $\frac{1}{k_{\s{R}}\wedge k_{\s{L}}}\ll\chi^2(p||q)\ll \frac{n^2}{k^2_{\s{R}} k^2_{\s{L}}}\wedge\frac{n}{k^2_{\s{R}}\vee k^2_{\s{L}}}$, for the sparse regime, while the detection problem can be solved by an exhaustive search, we do not have a polynomial-time algorithm. As mentioned above, we conjecture that, in fact, polynomial-time algorithms do not exist in this regions of parameters, and we next provide an evidence for this claim. To that end, we start with a brief introduction to the method of \emph{low-degree polynomials}. 

The basic premise of this method is that all polynomial-time algorithms for solving detection problems are captured by polynomials of low-degree. By now, there is growing strong evidence in support of this conjecture. The ideas below were first developed in a sequence of works in the sum-of-squares optimization literature \cite{barak2016nearly,Hopkins18,hopkins2017bayesian,hopkins2017power}. We follow the notations and definition in \cite{Hopkins18,Dmitriy19}. Any distribution $\pr_{\calH_0}$ on $\Omega_n=\{0,1\}^{\binom{n}{2}}$ induces an inner product of measurable functions $f,g:\Omega_n\to\mathbb{R}$ given by $\left\langle f,g \right\rangle_{\calH_0} = \bE_{\calH_0}[f(\s{G})g(\s{G})]$, and norm $\norm{f}_{\calH_0} = \left\langle f,f \right\rangle_{\calH_0}^{1/2}$. We Let $L^2(\pr_{\calH_0})$ denote the Hilbert space consisting of functions $f$ for which $\norm{f}_{\calH_0}<\infty$, endowed with the above inner product and norm. 
The underlying idea is to find the low-degree polynomial that best distinguishes $\pr_{\calH_0}$ from $\pr_{\calH_1}$ in the $L^2$ sense. Let $\calL_{n,\leq\s{D}}\subset L^2(\pr_{\calH_0})$ denote the linear subspace of polynomials $\Omega_n\to\mathbb{R}$ of degree at most $\s{D}\in\mathbb{N}$. 
Then, the \emph{$\s{D}$-low-degree likelihood ratio} $\s{L}_{n,\leq \s{D}}$ is the projection of a function $\s{L}_{n}$ to the span of coordinate-degree-$\s{D}$ functions, where the projection is orthogonal with respect to the inner product $\left\langle \cdot,\cdot \right\rangle_{\calH_0}$. As discussed above, the likelihood ratio optimally distinguishes $\pr_{\calH_0}$ from $\pr_{\calH_1}$ in the $L^2$ sense. The next lemma shows that over the set of low-degree polynomials, the $\s{D}$-low-degree likelihood ratio exhibits the same property.
\begin{lemma}[Optimally of $\s{L}_{n,\leq \s{D}}$ {\cite{hopkins2017bayesian,hopkins2017power,Dmitriy19}}]\label{lem:Dmitriy}
Consider the following optimization problem:
\begin{equation}
\begin{aligned}
\mathrm{max}
\;\bE_{\calH_1}f(\s{G})
\quad\mathrm{s.t.}
\quad\bE_{\calH_0}f^2(\s{G}) = 1,\; f\in\calL_{n,\leq\s{D}},
\end{aligned}\label{eqn:optimizationProblem}
\end{equation}
Then, the unique solution $f^\star$ for \eqref{eqn:optimizationProblem} is the $\s{D}$-low degree likelihood ratio $f^\star = \s{L}_{n,\leq \s{D}}/\norm{\s{L}_{n,\leq \s{D}}}_{\calH_0}$, and the value of the optimization problem is $\norm{\s{L}_{n,\leq \s{D}}}_{\calH_0}$. 
\end{lemma}
As was mentioned above, in the computationally-unbounded regime, an important property of the likelihood ratio is that if $\norm{\s{L}_n}_{\calH_0}$ is bounded then $\pr_{\calH_0}$ and $\pr_{\calH_1}$ are statistically indistinguishable. The following conjecture states that a computational analogue of this property holds, with $\s{L}_{n,\leq \s{D}}$ playing the role of the likelihood ratio. In fact it also postulates that polynomials of degree $\approx\log n$ are a proxy for polynomial-time algorithms. The conjecture below is based on \cite{Hopkins18,hopkins2017bayesian,hopkins2017power}, and \cite[Conj. 2.2.4]{Hopkins18}. We give an informal statement of this conjecture which appears in \cite[Conj. 1.16]{Dmitriy19}. For a precise statement, we refer the reader to, e.g., \cite[Conj. 2.2.4]{Hopkins18} and \cite[Sec. 4]{Dmitriy19}.
\begin{conjecture}[Low-degree conj., informal]\label{conj:1}
Given a sequence of probability measures $\pr_{\calH_0}$ and $\pr_{\calH_1}$, if there exists $\epsilon>0$ and $\s{D} = \s{D}(n)\geq (\log n)^{1+\epsilon}$, such that $\norm{\s{L}_{n,\leq \s{D}}}_{\calH_0}$ remains bounded as $n\to\infty$, then there is no polynomial-time algorithm that distinguishes $\pr_{\calH_0}$ and $\pr_{\calH_1}$.
\end{conjecture}
In the sequel, we will rely on Conjecture~\ref{conj:1} to give an evidence for the statistical-computational gap observed in our problem. At this point we would like to mention \cite[Hypothesis 2.1.5]{Hopkins18}, which states a more general form of Conjecture~\ref{conj:1} in the sense that it postulates that degree-$\s{D}$ polynomials are a proxy for $n^{O(D)}$-time algorithms. 

\begin{theorem}\label{thm:gap}
Consider the $\mathsf{PDBS}{\left(n,k_{\s{R}},k_{\s{L}},p,q\right)}$ detection problem in Problem~\ref{prob:pdbs}. 
\begin{enumerate}
    \item In the regime where $p,q=\Theta(1)$, if $(k^2_{\s{R}}\vee k^2_{\s{L}}) \ll n$, then $\norm{\s{L}_{n,\leq \s{D}}}_{\calH_0}\leq O(1)$, for any $\s{D}\leq C\log n$ and $C>0$. Conversely, there exists a positive integer $\s{D}$ such that if $(k^2_{\s{R}}\vee k^2_{\s{L}}) \gg n$, then $\norm{\s{L}_{n,\leq \s{D}}}_{\calH_0}\geq \omega(1)$.
    \item In the regime where $p,q=\Theta(n^{-\alpha})$, for some $\alpha\in(0,2]$, if $\frac{1}{(k_{\s{R}}\wedge k_{\s{L}})}\ll\chi^2(p||q)\ll \frac{n^2}{k^2_{\s{R}} k^2_{\s{L}}}\wedge\frac{n}{(k^2_{\s{R}}\vee k^2_{\s{L}})}$, then $\norm{\s{L}_{n,\leq \s{D}}}_{\calH_0}\leq O(1)$, for any $\s{D} = n^{o(1)}$. Conversely, there exists a positive integer $\s{D}$ such that if $\chi^2(p||q)\gg \frac{n^2}{k^2_{\s{R}} k^2_{\s{L}}}\wedge\frac{n}{(k^2_{\s{R}}\vee k^2_{\s{L}})}$, then $\norm{\s{L}_{n,\leq \s{D}}}_{\calH_0}\geq \omega(1)$.
\end{enumerate}
\end{theorem}
Together with Conjecture~\ref{conj:1}, Theorem~\ref{thm:gap} implies that if we take degree-$\log n$ polynomials as a proxy for all efficient algorithms, our calculations predict that an $n^{O(\log n)}$ algorithm does not exist when $(k^2_{\s{R}}\vee k^2_{\s{L}}) \ll n$ in the dense regime, and $\frac{1}{(k_{\s{R}}\wedge k_{\s{L}})}\ll\chi^2(p||q)\ll \frac{n^2}{k^2_{\s{R}} k^2_{\s{L}}}\wedge\frac{n}{(k^2_{\s{R}}\vee k^2_{\s{L}})}$ in the sparse regime. This is summarized in the following corollary.
\begin{corollary}[Computational lower bound]\label{cor:gap}
Consider the $\mathsf{PDBS}{\left(n,k_{\s{R}},k_{\s{L}},p,q\right)}$ detection problem in Problem~\ref{prob:pdbs}, and assume that Conjecture~\ref{conj:1} holds. A polynomial-time algorithm that achieves strong detection does not exist in the following situations:
\begin{enumerate}
    \item If $p,q=\Theta(1)$ and $(k^2_{\s{R}}\vee k^2_{\s{L}}) \ll n$.
    \item If $p,q = \Theta(n^{-\alpha})$, for some $\alpha\in(0,2]$, and $\frac{1}{(k_{\s{R}}\wedge k_{\s{L}})}\ll\chi^2(p||q)\ll \frac{n^2}{k^2_{\s{R}} k^2_{\s{L}}}\wedge\frac{n}{(k^2_{\s{R}}\vee k^2_{\s{L}})}$.
\end{enumerate}
\end{corollary}
These predictions agree precisely with the previously established statistical-computational tradeoffs in the previous subsections. A more explicit formula for the computational barrier which exhibits dependency on $\s{D}$ can be deduced from the proof of Theorem~\ref{thm:gap}; to keep the exposition simple we opted to present the refined result above. 

\begin{rmk}
    We would like to mention that while the focus in second part of Theorem~\ref{thm:gap} is on the specific scaling of $p,q=\Theta(n^{-\alpha})$, the result in Theorem~\ref{thm:gap} and its proof, hold also in the case where $p = \Theta(n^{-\alpha})$ and $q= \Theta(n^{-\beta})$, and $\alpha\leq\beta\leq 2\alpha$. This is known as the log-density regime.
\end{rmk}

\paragraph{Phase diagrams.} The statistical and computational barriers derived in Theorems~\ref{thm:sparse}--\ref{thm:gap} are best interpreted and understood using phase diagrams. Below, we consider the asymptotic regime where $k = k_{\s{R}}+k_{\s{L}}$ and $(p,q)$ are polynomials in~$n$, i.e., $k = \Theta(n^{\beta})$, for $\beta\in(0,1)$, and $p = c\cdot q = \Theta(n^{-\alpha})$, for $\alpha\in(0,2]$, and fixed $c>1$. We assume without any loss of generality that $k_{\s{R}} \geq k_{\s{L}}$. The phase diagram below include three important regions:
\begin{enumerate}
    \item \emph{Computationally easy regime (blue and green regions):} there are polynomial-time algorithms (count and degree tests) for the detection task.
    \item \emph{Computationally hard regime (red region):} there is an inefficient algorithm for detection (scan test), but the problem is computationally conjecturally hard (no polynomial-time algorithm exist) in the sense that the class of low-degree polynomials fails.
    \item \emph{Statistically impossible regime (gray region):} detection is statistically (or, information-theoretically) impossible.
\end{enumerate}
To obtain the following phase diagrams, we substitute $k_{\s{R}},k_{\s{L}} = \Theta(n^{\beta})$ and $\chi^2(p||q) = \Theta(n^{-\alpha})$ in \eqref{eqn:sparselower} and \eqref{eqn:sparseupper}, and Corollary~\ref{cor:gap}. It should be emphasized here that the phase diagrams below reflect the polynomial scale of the statistical and computational barriers we derive, as a function of the polynomial scale of the various parameters, i.e., $k_{\s{R}}$, $k_{\s{L}}$, $p$, and $q$, in $n$. As so, poly-log factors gaps between our lower and upper bounds are not captured in these diagrams. We draw the resulted phase transitions for various scaling regimes of $k_{\s{R}}$ and $k_{\s{L}}$:
\begin{itemize}
\item{Figure~\ref{fig:sparse bipartite 1} illustrates the ``balanced case", where $k_{\s{R}} = \Theta(k_{\s{L}})$, for which we have $m{\left({\mathsf{K}}_{\calR, \calL}\right)} = \Theta\left(k\right)$. Interestingly, we note that the obtained phase diagram is the same as for the $\mathsf{PDS}$ problem (e.g., \cite[Figure 1]{Hajek2015}), for which the maximal subgraph density behaves the same. Note that the degree test is redundant in this case. 

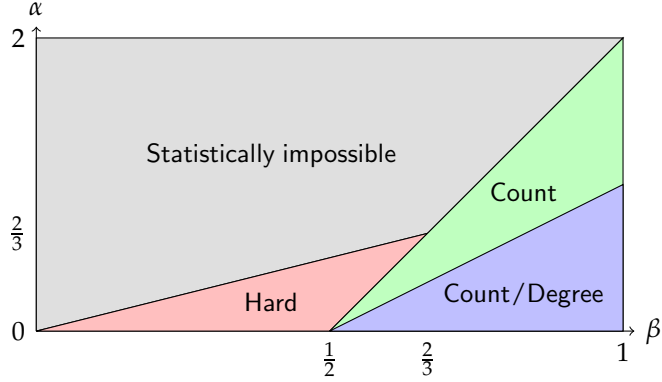
\begin{figure}[t!]
\centering

\begin{tikzpicture}[scale=1.5]
\tikzstyle{every node}=[font=\footnotesize]
\def\xUnit{5.2}
\def\yUnit{1.3}
\def\xmin{0}
\def\xmax{\xUnit + 0.1}
\def\ymin{0}
\def\ymax{\yUnit * 2 + 0.1}

\draw[->] (\xmin,\ymin) -- (\xmax,\ymin) node[right] {$\beta$};
\draw[->] (\xmin,\ymin) -- (\xmin,\ymax) node[above] {$\alpha$};

\node at (\xUnit, 0) [below] {$1$};
\node at (\xUnit * 0.667, 0) [below] {$\frac{2}{3}$};
\node at (\xUnit * 0.5, 0) [below] {$\frac{1}{2}$};
\node at (0, \yUnit * 2) [left] {$2$};
\node at (0, \yUnit * 0.667) [left] {$\frac{2}{3}$};
\node at (0, 0) [left] {$0$};

\filldraw[fill=gray!25, draw=black] (0, 0) -- (0, 2 * \yUnit) -- (1 * \xUnit, 2 * \yUnit) -- (0.667 * \xUnit, 0.667 * \yUnit) -- (0, 0);
\filldraw[fill=red!25, draw=black] (0, 0) -- (0.667 * \xUnit, 0.667 * \yUnit) -- (0.5 * \xUnit, 0) -- (0, 0);
\filldraw[fill=green!25, draw=black] (0.5 * \xUnit, 0) -- (1 * \xUnit, 2 * \yUnit) -- (1 * \xUnit, 0) -- (0.5 * \xUnit, 0);
\filldraw[fill=blue!25, draw=black] (0.5 * \xUnit, 0) -- (1 * \xUnit, 1 * \yUnit) -- (1 * \xUnit, 0) -- (0.5 * \xUnit, 0);

\node at (\xUnit * 0.4 , \yUnit * 1.2) {$\s{Statistically}\;\s{impossible}$};
\node at (\xUnit * 0.4 , \yUnit * 0.2) {$\s{Hard}$};
\node at (\xUnit * 0.83 , \yUnit * 0.95) {$\s{Count}$};
\node at (\xUnit * 0.83 , \yUnit * 0.25) {$\s{Count}/\s{Degree}$};

\end{tikzpicture}

\caption{Phase diagram as a function of $k= \Theta(n^{\beta})$ and $p,q=\Theta(n^{-\alpha})$, for the balanced case $k_{\s{R}} = \Theta(k_{\s{L}})$.}
\label{fig:sparse bipartite 1}
\end{figure}}

\item{Figure~\ref{fig:sparse bipartite 2} shows  the ``lightly unbalanced" case, where $k_{\s{R}} = \Theta{\left(k\right)}$, and $\omega(\sqrt{k}) \leq k_{\s{L}} \leq o\left(k\right)$. In this regime, $\omega(\sqrt{k}) \leq m{\left({\mathsf{K}}_{\calR, \calL}\right)} \leq o\left(k\right)$, and for illustration we take $k_{\s{L}} = \Theta(k^{\frac{2}{3}})$.

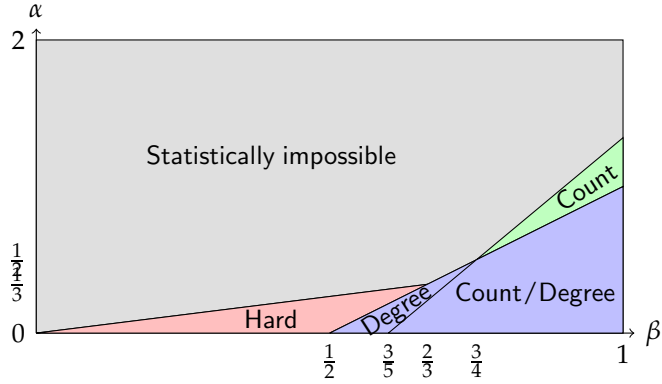
\begin{figure}[t!]
\centering

\begin{tikzpicture}[scale=1.5]
\tikzstyle{every node}=[font=\footnotesize]
\def\xUnit{5.2}
\def\yUnit{1.3}
\def\xmin{0}
\def\xmax{\xUnit + 0.1}
\def\ymin{0}
\def\ymax{\yUnit * 2 + 0.1}

\draw[->] (\xmin,\ymin) -- (\xmax,\ymin) node[right] {$\beta$};
\draw[->] (\xmin,\ymin) -- (\xmin,\ymax) node[above] {$\alpha$};

\node at (\xUnit, 0) [below] {$1$};
\node at (\xUnit * 0.75, 0) [below] {$\frac{3}{4}$};
\node at (\xUnit * 0.667, 0) [below] {$\frac{2}{3}$};
\node at (\xUnit * 0.5, 0) [below] {$\frac{1}{2}$};
\node at (\xUnit * 0.6, 0) [below] {$\frac{3}{5}$};
\node at (0, \yUnit * 2) [left] {$2$};
\node at (0, \yUnit * 0.5) [left] {$\frac{1}{2}$};
\node at (0, \yUnit * 0.333) [left] {$\frac{1}{3}$};
\node at (0, 0) [left] {$0$};

\filldraw[fill=gray!25, draw=black] (0, 0) -- (0, 2 * \yUnit) -- (1 * \xUnit, 2 * \yUnit) -- (1 * \xUnit, 1 * \yUnit) -- (0.667 * \xUnit, 0.333 * \yUnit) -- (0, 0);
\filldraw[fill=red!25, draw=black] (0, 0) -- (0.667 * \xUnit, 0.333 * \yUnit) -- (0.5 * \xUnit, 0) -- (0, 0);
\filldraw[fill=blue!25, draw=black] (0.5 * \xUnit, 0) -- (1 * \xUnit, 1 * \yUnit) -- (1 * \xUnit, 0) -- (0.5 * \xUnit, 0);
\filldraw[fill=green!25, draw=black] (0.75 * \xUnit, 0.5 * \yUnit) -- (1 * \xUnit, 1.333 * \yUnit) -- (1 * \xUnit, 1 * \yUnit) -- (0.75 * \xUnit, 0.5 * \yUnit);
\filldraw[fill=blue!25, draw=black] (0.6 * \xUnit, 0 * \yUnit) -- (0.75 * \xUnit, 0.5 * \yUnit) -- (1 * \xUnit, 1*\yUnit) -- (1 * \xUnit, 0 * \yUnit) -- (0.6 * \xUnit, 0 * \yUnit);

\node at (\xUnit * 0.4 , \yUnit * 1.2) {$\s{Statistically}\;\s{impossible}$};
\node at (\xUnit * 0.4 , \yUnit * 0.1) {$\s{Hard}$};
\node at (\xUnit * 0.85 , \yUnit * 0.27) {$\s{Count}/\s{Degree}$};
\node[rotate=33] at (\xUnit * 0.937, \yUnit * 1) {$\s{Count}$};
\node[rotate=33] at (\xUnit * 0.615, \yUnit * 0.15) {$\s{Degree}$};
\end{tikzpicture}

\caption{Phase diagram as a function of $k= \Theta(n^{\beta})$ and $p,q=\Theta(n^{-\alpha})$, for the lightly unbalanced case $k_{\s{R}} = \Theta{\left(k\right)}$ and $k_{\s{L}} = \Theta{(k^{\frac{2}{3}})}$.}
\label{fig:sparse bipartite 2}
\end{figure}}

\item{Figure~\ref{fig:sparse bipartite 3} shows the ``moderately unbalanced" case, where $k_{\s{R}} = \Theta{\left(k\right)}$, and $\omega(1) \leq k_{\s{L}} \leq O(\sqrt{k})$. In this regime, $\omega(1) \leq m{\left({\mathsf{K}}_{\calR, \calL}\right)} \leq O(\sqrt{k})$. Note that the count test is redundant, and for illustration we take $k_{\s{L}} = \Theta(k^{\frac{1}{2}})$.

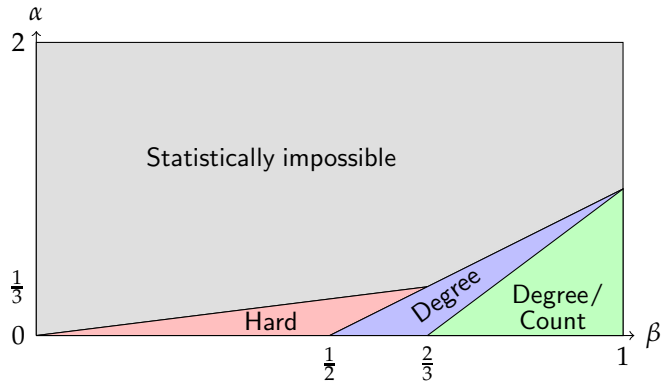
\begin{figure}[t!]
\centering

\begin{tikzpicture}[scale=1.5]
\tikzstyle{every node}=[font=\footnotesize]
\def\xUnit{5.2}
\def\yUnit{1.3}
\def\xmin{0}
\def\xmax{\xUnit + 0.1}
\def\ymin{0}
\def\ymax{\yUnit * 2 + 0.1}

\draw[->] (\xmin,\ymin) -- (\xmax,\ymin) node[right] {$\beta$};
\draw[->] (\xmin,\ymin) -- (\xmin,\ymax) node[above] {$\alpha$};

\node at (\xUnit, 0) [below] {$1$};
\node at (\xUnit * 0.667, 0) [below] {$\frac{2}{3}$};
\node at (\xUnit * 0.5, 0) [below] {$\frac{1}{2}$};
\node at (0, \yUnit * 2) [left] {$2$};
\node at (0, \yUnit * 0.333) [left] {$\frac{1}{3}$};
\node at (0, 0) [left] {$0$};

\filldraw[fill=gray!25, draw=black] (0, 0) -- (0, 2 * \yUnit) -- (1 * \xUnit, 2 * \yUnit) -- (1 * \xUnit, 1 * \yUnit) -- (0.667 * \xUnit, 0.333 * \yUnit) -- (0, 0);
\filldraw[fill=red!25, draw=black] (0, 0) -- (0.667 * \xUnit, 0.333 * \yUnit) -- (0.5 * \xUnit, 0) -- (0, 0);
\filldraw[fill=blue!25, draw=black] (0.5 * \xUnit, 0) -- (1 * \xUnit, 1 * \yUnit) -- (1 * \xUnit, 0) -- (0.5 * \xUnit, 0);
\filldraw[fill=green!25, draw=black] (2/3 * \xUnit, 0) -- (1 * \xUnit, 1 * \yUnit) -- (1 * \xUnit, 0) -- (2/3 * \xUnit, 0);

\node at (\xUnit * 0.4 , \yUnit * 1.2) {$\s{Statistically}\;\s{impossible}$};
\node at (\xUnit * 0.4 , \yUnit * 0.1) {$\s{Hard}$};
\node at (\xUnit * 0.89 , \yUnit * 0.27) {$\s{Degree}/$};
\node at (\xUnit * 0.88 , \yUnit * 0.11) {$\s{Count}$};
\node[rotate=36] at (\xUnit * 0.7 , \yUnit * 0.27) {$\s{Degree}$};

\end{tikzpicture}

\caption{Phase diagram as a function of $k= \Theta(n^{\beta})$ and $p,q=\Theta(n^{-\alpha})$, for the moderately unbalanced case $k_{\s{R}} = \Theta{\left(k\right)}$ and $\omega(1) \leq k_{\s{L}} \leq O(\sqrt{k})$, and for illustration we take $k_{\s{L}} = \Theta(k^{\frac{1}{2}})$.}
\label{fig:sparse bipartite 3}
\end{figure}}
\item{Finally, Figure~\ref{fig:sparse bipartite 4} shows the ``extremely unbalanced" case, where $k_{\s{R}} = \Theta{\left(k\right)}$ and $k_{\s{L}} = \Theta\left(1\right)$. In this regime, $m{\left({\mathsf{K}}_{\calR, \calL}\right)} = \Theta\left(1\right)$. This case is asymptotically equivalent to planting a star subgraph, which was studied in \cite{Laurent2019}. Note that both the count and the scan tests are redundant.

\begin{figure}[t!]
\centering

\begin{tikzpicture}[scale=1.5]
\tikzstyle{every node}=[font=\footnotesize]
\def\xUnit{5.2}
\def\yUnit{1.3}
\def\xmin{0}
\def\xmax{\xUnit + 0.1}
\def\ymin{0}
\def\ymax{\yUnit * 2 + 0.1}

\draw[->] (\xmin,\ymin) -- (\xmax,\ymin) node[right] {$\beta$};
\draw[->] (\xmin,\ymin) -- (\xmin,\ymax) node[above] {$\alpha$};

\node at (\xUnit, 0) [below] {$1$};
\node at (\xUnit * 0.5, 0) [below] {$\frac{1}{2}$};
\node at (0, \yUnit * 2) [left] {$2$};
\node at (0, 0) [left] {$0$};

\filldraw[fill=gray!25, draw=black] (0, 0) -- (0, 2 * \yUnit) -- (1 * \xUnit, 2 * \yUnit) -- (1 * \xUnit, 1 * \yUnit) -- (0.667 * \xUnit, 0.333 * \yUnit) -- (0.5 * \xUnit, 0) -- (0, 0);
\filldraw[fill=blue!25, draw=black] (0.5 * \xUnit, 0) -- (1 * \xUnit, 1 * \yUnit) -- (1 * \xUnit, 0) -- (0.5 * \xUnit, 0);

\node at (\xUnit * 0.4 , \yUnit * 1.2) {$\s{Statistically}\;\s{impossible}$};
\node at (\xUnit * 0.83 , \yUnit * 0.27) {$\s{Degree}$};

\end{tikzpicture}

\caption{Phase diagram as a function of $k= \Theta(n^{\beta})$ and $p,q=\Theta(n^{-\alpha})$, for the extremely unbalanced case $k_{\s{R}} = \Theta{\left(k\right)}$ and $k_{\s{L}} = \Theta\left(1\right)$.}
\label{fig:sparse bipartite 4}
\end{figure}
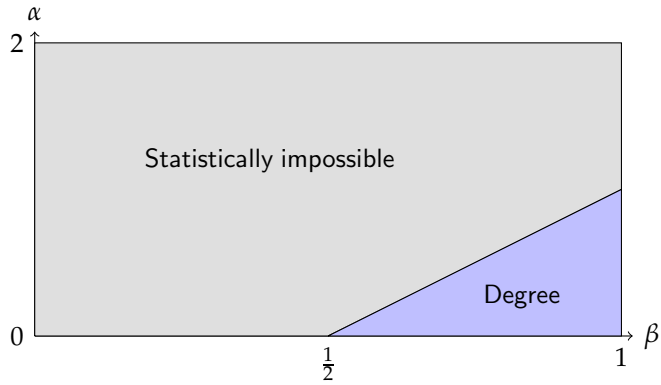}

\end{itemize}

\section{Proofs} \label{sec:proofs}
\subsection{Proof of Theorem \ref{thm:sparse}} \label{proof dense sparse lower bounds}
In this subsection, we prove Theorem~\ref{thm:sparse}, by finding conditions under which the risk of any algorithm strictly positive. To that end, we will use the following well-known fact that the optimal test, denoted by $\phi^\star_n$, minimizing the average risk, is given by,
\begin{align}
\phi^\star_n\left(\s{G}\right) \triangleq \begin{cases}
1,\ &\text{if }\mathsf{L}_n\left(\s{G}\right) \geq 1\\
0,\ &\text{otherwise},   
\end{cases}
\end{align}
where $\mathsf{L}_n\left(\s{G}\right) \triangleq \frac{\pr_{\calH_1}\left(\s{G}\right)}{\pr_{\calH_0}\left(\s{G}\right)}$. 
The optimal risk $\mathsf{R}_n^\star$ associated with the above optimal test, can be lower bounded using the Cauchy–Schwarz inequality as follows,
\begin{align}
\mathsf{R}_n^\star &= 1-d_{\mathsf{TV}}\left(\pr_{\calH_0}, \pr_{\calH_1}\right)\\
&=1-\frac{1}{2}\mathbb{E}_{\calH_0}\vert\mathsf{L}_n\left(\s{G}\right)-1\vert   \\
&\geq1-\frac{1} {2}\sqrt{\mathbb{E}_{\calH_0}[\mathsf{L}_n^2\left(\s{G}\right)]-1}.\label{risk lower bound}
\end{align}
Thus, for finding conditions under which the optimal risk is positive we can find conditions under which $\mathbb{E}_{\calH_0}[\mathsf{L}_n^2\left(\s{G}\right)]$ lies in the interval $(1,5)$. 

We start by finding a formula for the second moment of the likelihood under $\calH_0$. Note that the likelihood ratio is given by
\begin{align}
    \mathsf{L}_n(\s{G}) = \bE_{\mathsf{K}_{\calR, \calL}}\pp{\frac{\pr_{\calH_1}(\s{G}\vert\mathsf{K}_{\calR, \calL})}{\pr_{\calH_0}(\s{G})}},
\end{align}
where the expectation is taken w.r.t. $\mathsf{K}_{\calR, \calL}$ that is uniformly distributed over $\calK_{n,k_{\s{R}
,k_{\s{L}}}}$, and $\pr_{\calH_1}(\s{G}\vert\mathsf{K}_{\calR, \calL})$ denote the conditional distribution of $\s{G}$ given that $\mathsf{K}_{\calR, \calL}$ is the planted bipartite graph under $\calH_1$. Introducing an independent copy $\mathsf{K}_{\calR, \calL}'$ of $\mathsf{K}_{\calR, \calL}$, the squared likelihood ratio can be expressed as
\begin{align}
    \mathsf{L}^2_n(\s{G}) = \bE_{\mathsf{K}_{\calR, \calL}\indep\mathsf{K}_{\calR, \calL}'}\pp{\frac{\pr_{\calH_1}(\s{G}\vert\mathsf{K}_{\calR, \calL})\pr_{\calH_1}(\s{G}\vert\mathsf{K}'_{\calR, \calL})}{\pr_{\calH_0}^2(\s{G})}},
\end{align}
where $\mathsf{K}_{\calR, \calL}\indep\mathsf{K}_{\calR, \calL}'$ denotes that $\mathsf{K}_{\calR, \calL}$ and $\mathsf{K}'_{\calR, \calL}$ are independent. Interchanging the expectations yields,
\begin{align}
\mathbb{E}_{\calH_0}[\mathsf{L}_n^2\left(\s{G}\right)] &= \bE_{\mathsf{K}_{\calR, \calL}\indep\mathsf{K}_{\calR, \calL}'}\pp{\bE_{\calH_0}\p{\frac{\pr_{\calH_1}(\s{G}\vert\mathsf{K}_{\calR, \calL})\pr_{\calH_1}(\s{G}\vert\mathsf{K}'_{\calR, \calL})}{\pr_{\calH_0}^2(\s{G})}}}.
\end{align}
Let $\s{A}$ denote the adjacency matrix of $\s{G}$, and for any pair $(i,j)\in\binom{\left[n\right]}{2}$, we define $f(\s{A}_{ij}) \triangleq \p{p/q}^{\s{A}_{ij}}\pp{(1-p)/(1-q)}^{\s{A}_{ij}}$. Then, note that $\bE_{\calH_0}[f(\s{A}_{ij})]=1$ and $\bE_{\calH_0}[f^2(\s{A}_{ij})]=1+\chi^2(p||q)$. Thus, 
\begin{align}
&\mathbb{E}_{\calH_0}[\mathsf{L}_n^2\left(\s{G}\right)]\nonumber\\
&\quad=\mathbb{E}_{\mathsf{K}_{\calR, \calL}\indep\mathsf{K}_{\calR, \calL}'}\left[\prod_{(i,j) \in {\mathsf{K}}_{\calR, \calL} \cap \mathsf{K}'_{\calR, \calL}}\mathbb{E}_{\calH_0}\left[f^2\left(\s{A}_{ij}\right)\right] \cdot \prod_{(i,j) \in {\mathsf{K}}_{\calR, \calL} \triangle \mathsf{K}'_{\calR, \calL}}\mathbb{E}_{\calH_0}\left[f\left(\s{A}_{ij}\right)\right] \right]    \\ 
&\quad=\mathbb{E}_{\mathsf{K}_{\calR, \calL}\indep\mathsf{K}_{\calR, \calL}'}\left[\left(1+\chi^2(p||q)\right)^{\vert E\left({\mathsf{K}}_{\calR, \calL} \cap \mathsf{K}'_{\calR, \calL}
\right)\vert}\right] \label{L square equal} \\
&\quad\leq \mathbb{E}_{\mathsf{K}_{\calR, \calL}\indep\mathsf{K}_{\calR, \calL}'}{\left[\exp\p{\chi^2(p||q)\cdot {\vert E\left({\mathsf{K}}_{\calR, \calL} \cap \mathsf{K}'_{\calR, \calL} \right)\vert}}\right]},\label{L square bound}
\end{align}
where in the inequality follows form the fact that $1+x\leq \exp(x)$, for $x\geq0$. Note that \eqref{L square bound} is simply the moment generating function of the intersection between two random copies of ${\mathsf{K}}_{\calR, \calL}$. Our main observation here is the fact that
\begin{align}
\vert E\left({\mathsf{K}}_{\calR, \calL} \cap \mathsf{K}'_{\calR, \calL}\right)\vert = \vert \calR \cap \calR'\vert \cdot \vert \calL \cap \calL'\vert + \vert \calR \cap \calL'\vert \cdot \vert \calL \cap \calR'\vert,
\end{align}
which follows directly from the structure of bipartite graphs. Then, by Cauchy–Schwartz inequality,
\begin{align}
\mathbb{E}^2{\left[e^{\chi^2(p||q) \vert E\left({\mathsf{K}}_{\calR, \calL} \cap \mathsf{K}'_{\calR, \calL}\right)\vert}\right]} &= \mathbb{E}^2{\left[e^{\chi^2(p||q) \left( \vert \calR \cap \calR'\vert \cdot \vert \calL \cap \calL'\vert + \vert \calR \cap \calL'\vert \cdot \vert \calL \cap \calR'\vert \right)}\right]} \\
&\leq \mathbb{E}{\left[e^{2 \chi^2(p||q) \vert \calR \cap \calR'\vert \cdot \vert \calL \cap \calL'\vert}\right]} \cdot \mathbb{E}{\left[e^{2\chi^2(p||q) \vert \calR \cap \calL'\vert \cdot \vert \calL \cap \calR'\vert}\right]} \label{Cauchy–Schwarz Sparse Full}.
\end{align}
We next bound the first term at the right-hand-side of \eqref{Cauchy–Schwarz Sparse Full}. Since $\calR$ and $\calR'$ are drawn uniformly at random, given $H\triangleq\vert \calL \cap \calL'\vert$, we have,
\begin{align}
\vert \calR \cap \calR'\vert \preceq \mathsf{Hypergeometric}\left(n-H,k_{\s{R}},k_{\s{R}}\right),  
\end{align}
where $X \preceq Y$ if the random variable $X$ is stochastically dominated by $Y$. Now it is well-known that (see, e.g., \cite[Lemma 5]{Castro2014}) $\mathsf{Hypergeometric}\left(n-H,k_{\s{R}},k_{\s{R}}\right)\preceq \mathsf{Binomial}\left(k_{\s{R}},\frac{k_{\s{R}}}{n-H}\right)$. Since we assume that $k_{\s{R}} + k_{\s{L}} = o\left(n\right)$, and as so $k_{\s{R}},k_{\s{L}}\leq n/2$, which in turn implies that $H \leq k_{\s{L}}\leq n/2$, we further have that $\mathsf{Binomial}\left(k_{\s{R}},\frac{k_{\s{R}}}{n-H}\right)\preceq\mathsf{Binomial}\left(k_{\s{R}},\frac{2k_{\s{R}}}{n}\right)$. Therefore, by the law of total probability we get,
\begin{align} \label{after CS start}
\mathbb{E}{\left[e^{2 \chi^2(p||q) \vert \calR \cap \calR'\vert \cdot \vert \calL \cap \calL'\vert}\right]} &= \mathbb{E}{\left[\mathbb{E}{\left.\left[e^{2 \chi^2(p||q) \vert \calR \cap \calR'\vert \cdot \vert \calL \cap \calL'\vert} \right|  \vert \calL \cap \calL'\vert = H\right]}\right]} \\
&\leq \mathbb{E}{\left[e^{2 \chi^2(p||q) \mathsf{Binomial}\left(k_{\s{R}},\frac{2k_{\s{R}}}{n}\right) \cdot \vert \calL \cap \calL'\vert} \right]} \\
&\leq \mathbb{E}{\left[e^{2 \chi^2(p||q) \mathsf{Binomial}\left(k_{\s{R}},\frac{2k_{\s{R}}}{n}\right) \cdot \mathsf{Binomial}\left(k_{\s{L}},\frac{2k_{\s{L}}}{n}\right)} \right]} \\
&\triangleq \mathbb{E}{\left[e^{2 \chi^2(p||q) \s{B} \cdot \s{B}'} \right]}, \label{MGF B1 B2}
\end{align}
where $\s{B} \triangleq \mathsf{Binomial}\left(k_{\s{R}},\frac{2k_{\s{R}}}{n}\right)$ and $\s{B}' \triangleq \mathsf{Binomial}\left(k_{\s{L}},\cdot \frac{2k_{\s{L}}}{n}\right)$. In the same way, the second term at the right-hand-side of \eqref{Cauchy–Schwarz Sparse Full} can be bounded as,
\begin{align}
\mathbb{E}{\left[e^{2 \chi^2(p||q) \vert \calR \cap \calL'\vert \cdot \vert \calL \cap \calR'\vert}\right]} \leq \mathbb{E}{\left[e^{2 \chi^2(p||q) \bar{\s{B}} \cdot \bar{\s{B}}'} \right]},\label{MGF B3 B4}
\end{align}
where $\bar{\s{B}} \triangleq \mathsf{Binomial}\left(k_{\s{R}},\frac{2k_{\s{L}}}{n}\right)$ and $\bar{\s{B}}' \triangleq \mathsf{Binomial}\left(k_{\s{L}},\frac{2k_{\s{R}}}{n}\right)$. Our goal is now to find conditions under which each one of the  expectation terms in \eqref{Cauchy–Schwarz Sparse Full} is upper bounded by, for example, 2. To that end, we will use the upper bounds we derived in \eqref{MGF B1 B2} and \eqref{MGF B3 B4}. These in turn will imply that \eqref{L square bound}, and as so, the second moment, is upper bounded by 2 as well, which due to \eqref{risk lower bound} entail the impossibility of strong detection. 

We start by upper bounding \eqref{MGF B1 B2} using two different approaches. The first approach will lead to the condition in \eqref{eqn:densecond1}, while the second approach will lead to the condition in \eqref{eqn:sparsecond2}. Accordingly, under \emph{either} one of those conditions \eqref{MGF B1 B2} will be bounded by 2. Let us start with the first approach. Using the fact that $\s{B}' \leq k_{\s{L}}$ with probability one, we have,
\begin{align}
\mathbb{E}{\left[e^{2 \chi^2(p||q) \s{B} \cdot \s{B}'} \right]} &\leq \mathbb{E}{\left[e^{2 \chi^2(p||q) \s{B} \cdot k_{\s{L}}} \right]}\label{eqn:B2up}\\
&={\left[1+2\frac{k_{\s{R}}}{n} \left(e^{2 \chi^2(p||q) k_{\s{L}}}-1\right)\right]}^{k_{\s{R}}},\label{eqn:by2}
\end{align}
where we have used the moment generating function of the binomial random variable $\s{B}$. Now, since $2^{\frac{1}{k_{\s{R}}}}-1 \geq \frac{\log{2}}{k_{\s{R}}}$, the right-hand-side of \eqref{eqn:by2} is upper bounded by 2 if
\begin{align}
2 \frac{k_{\s{R}}}{n} \left(e^{2 \chi^2(p||q) k_{\s{L}}}-1\right) \leq \frac{\log{2}}{k_{\s{R}}},
\end{align}
which is equivalent to
\begin{align}
\chi^2(p||q)\leq \frac{1}{2k_{\s{L}}}\log{\left(1+\frac{n}{{k^2_{\s{R}}}} \frac{\log{2}}{2}\right)}. \label{dense full cond}
\end{align}
If we repeat the same arguments as above but use the bound $\s{B}\leq k_{\s{R}}$ in \eqref{eqn:B2up} instead, we will get that if 
\begin{align}
\chi^2(p||q)\leq \frac{1}{2k_{\s{R}}}\log{\left(1+\frac{n}{{k^2_{\s{L}}}} \frac{\log{2}}{2}\right)},\label{dense full cond2}
\end{align}
then the right-hand-side of \eqref{MGF B1 B2} is upper bounded by 2. Combining both conditions we get,
\begin{align}
    \chi^2(p||q)\leq \max\ppp{\frac{1}{2k_{\s{L}}}\log{\left(1+\frac{n}{{k^2_{\s{R}}}} \frac{\log{2}}{2}\right)},\frac{1}{2k_{\s{R}}}\log{\left(1+\frac{n}{{k^2_{\s{L}}}} \frac{\log{2}}{2}\right)}},
\end{align}
which is exactly the condition in \eqref{eqn:densecond1}. 

Next, we propose another approach for bounding \eqref{MGF B1 B2}. Using the moment generating function of the binomial random variable $\s{B}$, we obtain,
\begin{align}
\mathbb{E}{\left[e^{2 \chi^2(p||q) \s{B} \cdot \s{B}'} \right]} &= \mathbb{E}{\left[\left(1+2\frac{k_{\s{R}}}{n} \left(e^{ 2 \chi^2(p||q) \s{B}'}-1\right)\right)^{k_{\s{R}}}\right]}.
\end{align}
We assume that $2\chi^2(p||q) k_{\s{L}} \leq 1$, and thus, combined with the fact that $\s{B}' \leq k_{\s{L}}$ with probability one, we have $2\chi^2(p||q) \s{B}' \leq 1$ with probability one. Therefore, using the fact that $\exp(x) - 1\leq x + x^2$, for $x < 1.79$, we obtain
\begin{align}
&\mathbb{E}{\left[\left(1+2\frac{k_{\s{R}}}{n} \left(e^{ 2 \chi^2(p||q) \s{B}'}-1\right)\right)^{k_{\s{R}}}\right]} \nonumber\\
&\hspace{2.5cm}\leq \mathbb{E}{\left[\left(1+2 \frac{k_{\s{R}}}{n} \left(2 \chi^2(p||q) \s{B}' + \left(2 \chi^2(p||q) \s{B}'\right)^2\right)\right)^{k_{\s{R}}}\right]} \\ 
&\hspace{2.5cm}\leq \mathbb{E}{\left[\left(1+8\frac{k_{\s{R}}}{n} \chi^2(p||q) \s{B}'\right)^{k_{\s{R}}}\right]}\\
&\hspace{2.5cm}\leq \mathbb{E}{\left[\exp\p{8 \frac{k_{\s{R}}^2}{n}\chi^2(p||q) \s{B}'}\right]}\\
&\hspace{2.5cm}=\left[1+2\frac{k_{\s{L}}}{n}\left(e^{8 \frac{k_{\s{R}}^2}{n} \chi^2(p||q)}-1\right)\right]^{k_{\s{L}}},\label{after two MGFs}
\end{align}
where the second inequality is because $2 \chi^2(p||q) \s{B}'\leq1$, and so $\left(2 \chi^2(p||q) \s{B}'\right)^2\leq 2 \chi^2(p||q) \s{B}'$, with probability one, the third inequality is due to the fact that $1+x\leq\exp(x)$, for any $x\geq0$, and in the last equality we have used the moment generating function of the binomial random variable $\s{B}'$. As in the previous subsection, using the inequality $2^{\frac{1}{k_{\s{L}}}}-1 \geq \frac{\log{2}}{k_{\s{L}}}$, it can be seen that \eqref{after two MGFs}, as as so \eqref{MGF B1 B2} as well, is upper bounded by 2, if
\begin{align} 
\chi^2(p||q)\leq\frac{n}{8k_{\s{R}}^2} \cdot {\log{\left(1+\frac{n}{k_{\s{L}}^2}  \frac{\log 2}{2}\right)}}.
\end{align}
Combined with our assumption that $2\chi^2(p||q) k_{\s{L}} \leq 1$, we obtain that \eqref{MGF B1 B2} is upper bounded by 2 if,
\begin{align}\label{after CS first term}
   \chi^2(p||q)\leq\min\ppp{\frac{1}{2k_{\s{L}}},\frac{n}{8k_{\s{R}}^2} \cdot {\log{\left(1+\frac{n}{k_{\s{L}}^2}  \frac{\log 2}{2}\right)}}}.
\end{align}
If we repeat the same arguments as above, but replace the roles of $\s{B}$ and $\s{B}'$, and assume that $2\chi^2(p||q) k_{\s{R}} \leq 1$ instead, we will get that \eqref{MGF B1 B2} is upper bounded by 2, if
\begin{align} \label{after CS first term2}
\chi^2(p||q)\leq\min\ppp{\frac{1}{2k_{\s{R}}},\frac{n}{8k_{\s{L}}^2} \cdot {\log{\left(1+\frac{n}{k_{\s{R}}^2}  \frac{\log 2}{2}\right)}}},
\end{align}
and so when combined, \eqref{MGF B1 B2} is upper bounded by 2 if either \eqref{after CS first term} or \eqref{after CS first term2} hold, which is exactly the condition stated in \eqref{eqn:sparsecond2}.

Finally, we upper bound \eqref{MGF B3 B4}. Repeating the exact same arguments as in the later approach above, it can be shown that the right-hand-side of \eqref{MGF B3 B4} is upper bounded by 2 if,
\begin{align} \label{after CS second term}
\chi^2(p||q)\leq \frac{n}{8k_{\s{R}} k_{\s{L}}} \cdot {\log{\left(1+\frac{n}{k_{\s{R}} k_{\s{L}}}  \frac{\log 2}{2}\right)}},
\end{align}
which is the condition we stated in \eqref{eqn:sparsecond1}. This concludes the proof.

\subsection{Proof of Theorem~\ref{thm:pdbs}}

In this subsection, we prove Theorem~\ref{thm:pdbs} by upper bounding the risk associated with each of the three tests proposed in Section~\ref{sec:main}, starting with the scan test.
\subsubsection{Scan test}
Under the null hypothesis, for any $(\calR',\calL')\in\calS_{k_{\s{R}},k_{\s{L}}}$ we have that $\sum_{\substack{i \in \calR', j \in \calL'}}\s{A}_{ij} \sim \mathsf{Binomial}{\left(k_{\s{R}} k_{\s{L}}, q\right)}$. Recall the definition of $\mathsf{Scan}{\left(\s{A}\right)}$ in \eqref{algo:scan non-complete}, and note that the maximum in \eqref{algo:scan non-complete} taken over a set of size $|\calS_{k_{\s{R}},k_{\s{L}}}| = \binom{n}{k_{\s{R}}}\binom{n-k_{\s{R}}}{k_{\s{L}}}$, which can be upper bounded as follows,
\begin{align}
\binom{n}{k_{\s{R}}}\binom{n-k_{\s{R}}}{k_{\s{L}}} \leq \binom{n}{k_{\s{R}}}\binom{n}{k_{\s{L}}} \leq \left(\frac{en}{k_{\s{R}}}\right)^{k_{\s{R}}} \left(\frac{en}{k_{\s{L}}}\right)^{k_{\s{L}}} \leq \left(\frac{en}{k_{\s{R}}\vee k_{\s{L}} }\right)^{2 k_{\s{R}}\vee k_{\s{L}}},
\end{align}
where the second inequality follows from $(y/x)^x\leq\binom{y}{x}\leq (ey/x)^x$, and the last inequality is because $(n/a)^a$ is monotonically non-decreasing as a function of $a$ if $a=o\left(n\right)$. Thus, by the union bound and Bernstein's inequality we have,
\begin{align}
    \pr_{\calH_0}\p{\phi_{\s{Scan}}(\s{A})=1}&=\pr_{\calH_0}{\left(\mathsf{Scan}{\left(\s{A}\right)} \geq \tau_{\mathsf{Scan}}\right)}\\
    &\leq \sum_{(\calR',\calL')\in\calS_{k_{\s{R}},k_{\s{L}}}}{\pr_{\calH_0}{\left[\sum_{\substack{i \in \calR', j \in \calL'}}\s{A}_{ij} \geq \tau_{\mathsf{Scan}}\right]}}    \\
    &\leq \left(\frac{en}{k_{\s{R}}\vee k_{\s{L}} }\right)^{2 k_{\s{R}}\vee k_{\s{L}}} \cdot\exp{\left(-
    \frac{k^2_{\s{R}} k^2_{\s{L}} \cdot {\left(p-q\right)}^2/4}{2k_{\s{R}} k_{\s{L}} q+k_{\s{R}} k_{\s{L}}\cdot{\left(p-q\right)}/3}\right)}    \\
    &\leq \exp{\left(k_{\s{R}}\vee k_{\s{L}} \cdot \log{\left(\frac{en}{k_{\s{R}}\vee k_{\s{L}}}\right)}-C_1\cdot k_{\s{R}} k_{\s{L}} \cdot \frac{{\left(p-q\right)}^2}{q\left(1-q\right)}\right)},\label{eqn:scanType1}
\end{align}
where $C_1>0$, and the last inequality follows from the assumptions that $|p-q|=O(q)$ and $q<c<1$, for some $0<c<1$. Under the alternative hypothesis, if ${\mathsf{K}}^\star_{\calR,\calL}$ denotes the underlying planted bipartite subgraph, then there exist right and left sets of vertices $(\calR^{\ast}, \calL^{\ast})$ such that $\sum_{\substack{i \in \calR^{\ast}, j \in \calL^{\ast}}}\s{A}_{ij} \sim \mathsf{Binomial}{\left(k_{\s{R}} k_{\s{L}}, p\right)}$. Therefore, by the multiplicative Chernoff's bound, we have,
\begin{align}
    \pr_{\calH_1}\p{\phi_{\s{Scan}}(\s{A})=0}&=\pr_{\calH_1}{\left(\mathsf{Scan}{\left(\s{A}\right)} \leq \tau_{\mathsf{Scan}}\right)}\\
    &\leq \pr_{\calH_1}{\left[\sum_{\substack{i \in \calR^{\ast}, j \in \calL^{\ast}}} \s{A}_{ij} \leq \tau_{\mathsf{Scan}}\right]}    \\
    &\leq \exp{\left(-
    \frac{k^2_{\s{R}} k^2_{\s{L}} \cdot {\left(p-q\right)}^2/4}{2k_{\s{R}} k_{\s{L}} p}\right)}   \\
    &\leq \exp{\left(-C_2\cdot k_{\s{R}} k_{\s{L}} \cdot \frac{{\left(p-q\right)}^2}{q\left(1-q\right)}\right)},\label{eqn:scanType2}
\end{align}
for some constant $C_2> 0$, and again the last inequality follows from the assumptions that $|p-q|=O(q)$ and $q<c<1$. Thus, from \eqref{eqn:scanType1} and \eqref{eqn:scanType2}, it is clear that there exists a constant $C_3>0$ such that,
\begin{align}
\s{R}_n(\phi_{\s{Scan}}) \leq 2 \cdot \exp{\left(k_{\s{R}}\vee k_{\s{L}} \cdot \log{\left(\frac{n}{k_{\s{R}}\vee k_{\s{L}}}\right)}-C_3\cdot k_{\s{R}} k_{\s{L}} \cdot \frac{{\left(p-q\right)}^2}{q\left(1-q\right)}\right)},
\end{align}
which is less than $\delta$ provided that,
\begin{align}
\chi^2(p||q)=\frac{{\left(p-q\right)}^2}{q\left(1-q\right)} &\geq C {\left[\frac{k_{\s{R}}\vee k_{\s{L}}}{k_{\s{R}} k_{\s{L}}} \cdot \log{\left(\frac{n}{k_{\s{R}}\vee k_{\s{L}}}\right)}+\frac{\log{\frac{2} {\delta}}}{k_{\s{R}} k_{\s{L}}}\right]}\\&=C
{\left[\frac{\log{\left(\frac{n}{k_{\s{R}}\vee k_{\s{L}}}\right)}+\frac{\log{\frac{2} {\delta}}}{k_{\s{R}}\vee k_{\s{L}}}}{k_{\s{R}}\wedge k_{\s{L}}}\right]},
\end{align}
for some $C>0$, which completes the proof.
\subsubsection{Count test}
Recall the count statistics in \eqref{algo:count}. Under the null hypothesis, it is clear that $\mathsf{Count}{\left(\s{A}\right)} \sim \mathsf{Binomial}{\left(\binom{n}{2}, q\right)}$. Therefore, by Bernstein's inequality, 
\begin{align}
    \pr_{\calH_0}\p{\phi_{\s{Count}}(\s{A})=1}&=\pr_{\calH_0}{\left(\mathsf{Count}{\left(\s{A}\right)} \geq \tau_{\mathsf{Count}}\right)} 
    \\
    &\leq \exp{\left(-
    \frac{k^2_{\s{R}} k^2_{\s{L}} \cdot {\left(p-q\right)}^2/4}{2{\binom{n}{2}}q+k_{\s{R}} k_{\s{L}}\cdot{\left(p-q\right)}/3}\right)}    \\
    &\leq \exp{\left(-C_1\cdot \frac{k^2_{\s{R}} k^2_{\s{L}}}{n^2} \cdot \frac{{\left(p-q\right)}^2}{q\left(1-q\right)}\right)},\label{eqn:countType1}
\end{align}
where $C_1>0$, and the last inequality follows from the assumptions that $|p-q|=O(q)$ and $q<c<1$, for some $0<c<1$. Under the alternative hypothesis, the random variable $\mathsf{Count}{\left(\s{A}\right)}$ is distributed as the independent sum of $\mathsf{Binomial}{\left(k_{\s{R}} k_{\s{L}}, p\right)}$ and $\mathsf{Binomial}{\left(\binom{n}{2} - k_{\s{R}} k_{\s{L}}, q\right)}$. Therefore, by the multiplicative Chernoff's bound,
\begin{align}
    \pr_{\calH_1}\p{\phi_{\s{Count}}(\s{A})=0}&=\pr_{\calH_1}{\left(\mathsf{Count}{\left(\s{A}\right)} \leq \tau_{\mathsf{Count}}\right)} 
    \\
    &\leq \exp{\left(-
    \frac{k^2_{\s{R}} k^2_{\s{L}} \cdot {\left(p-q\right)}^2/4}{2{\binom{n}{2}}q+2\cdot k_{\s{R}} k_{\s{L}}\cdot{\left(p-q\right)}}\right)}    \\
    &\leq \exp{\left(-C_2\cdot \frac{k^2_{\s{R}} k^2_{\s{L}}}{n^2} \cdot \frac{{\left(p-q\right)}^2}{q\left(1-q\right)}\right)},\label{eqn:countType2}
\end{align}
where $C_2>0$. Thus, combining \eqref{eqn:countType1} and \eqref{eqn:countType2}, we obtain that there exists a constant $C_3>0$, such that,
\begin{align}
\s{R}_n(\phi_{\s{Count}}) \leq 2 \cdot \exp{\left(-C_3\cdot \frac{k^2_{\s{R}} k^2_{\s{L}}}{n^2} \cdot \frac{{\left(p-q\right)}^2}{q\left(1-q\right)}\right)}.
\end{align}
This is less than $\delta$ provided that,
\begin{align}
\chi^2(p||q)=\frac{{\left(p-q\right)}^2}{q\left(1-q\right)} \geq \frac{n^2}{C_3\cdot k^2_{\s{R}} k^2_{\s{L}}} \cdot \log{\frac{2}{\delta}},
\end{align}
which completes the proof.
\subsubsection{Degree test}
Recall the degree statistics in \eqref{algo:degree}. Under the null hypothesis, we clearly have that $\sum_{j} \s{A}_{ij} \sim \mathsf{Binomial}{\left(n-1,q\right)}$. Therefore, by the union bound and Bernstein's inequality, we get,
\begin{align}
    \pr_{\calH_0}\p{\phi_{\s{Deg}}(\s{A})=1}&=\pr_{\calH_0}{\left(\mathsf{MaxDeg}{\left(\s{A}\right)} \geq \tau_{\mathsf{Deg}}\right)} 
    \\
    &\leq n \cdot \exp{\left(-
    \frac{{\left(k_{\s{R}}\vee k_{\s{L}}\right)}^2 \cdot {\left(p-q\right)}^2/4}{2\left(n-1\right)q+k_{\s{R}}\vee k_{\s{L}}\cdot{\left(p-q\right)}/3}\right)}    \\
    &\leq \exp{\left(\log{n}-C_1\cdot \frac{{\left(k_{\s{R}}\vee k_{\s{L}}\right)}^2}n \cdot \frac{{\left(p-q\right)}^2}{q\left(1-q\right)}\right)},\label{eqn:degreeType1}
\end{align}
where $C_1>0$, and the last inequality follows from the assumptions that $|p-q|=O(q)$ and $q<c<1$, for some $0<c<1$. Under the alternative hypothesis, there is at least one row $i^\ast$ such that $\sum_{j} \s{A}_{i^\ast j}$ can be represented as the independent sum of $\mathsf{Binomial}{\left(k_{\s{R}}\vee k_{\s{L}}, p\right)}$ and $\mathsf{Binomial}{\left(n-1-k_{\s{R}}\vee k_{\s{L}}, q\right)}$. Therefore, by the multiplicative Chernoff's bound, we obtain,
\begin{align}
    \pr_{\calH_1}\p{\phi_{\s{Deg}}(\s{A})=0}&=\pr_{\calH_1}{\left(\mathsf{MaxDeg}{\left(\s{A}\right)} \leq \tau_{\mathsf{Deg}}\right)} 
    \\
    &\leq \exp{\left(-
    \frac{{\left(k_{\s{R}}\vee k_{\s{L}}\right)}^2 \cdot {\left(p-q\right)}^2/4}{2\left(n-1\right)q+2\cdot k_{\s{R}}\vee k_{\s{L}}\cdot{\left(p-q\right)}}\right)}    \\
    &\leq \exp{\left(-C_2\cdot \frac{{\left(k_{\s{R}}\vee k_{\s{L}}\right)}^2}n \cdot \frac{{\left(p-q\right)}^2}{q\left(1-q\right)}\right)},\label{eqn:degreeType2}
\end{align}
for some constant $C_2 > 0$. Thus, combining \eqref{eqn:degreeType1} and \eqref{eqn:degreeType2}, we obtain,
\begin{align}
\s{R}_n(\phi_{\s{Deg}}) \leq 2 \cdot \exp{\left(\log{n}-C\cdot \frac{{\left(k_{\s{R}}\vee k_{\s{L}}\right)}^2}n \cdot \frac{{\left(p-q\right)}^2}{q\left(1-q\right)}\right)},
\end{align}
for some constant $C>0$. This is less than $\delta$ provided that,
\begin{align}
\chi^2(p||q)=\frac{{\left(p-q\right)}^2}{q\left(1-q\right)} \geq \frac{n}{C\cdot {\left(k_{\s{R}}\vee k_{\s{L}}\right)}^2} \cdot \left(\log n + \log{\frac{2}{\delta}}\right),
\end{align}
which completes the proof.

\subsection{Proof of Theorem~\ref{thm:gap}}

Let us expand the likelihood ratio $\s{L}_n$ in a basis of orthogonal polynomials with respect to $\pr_{\calH_0}$. Specifically, suppose that $f_0,f_1,\ldots,f_m:\Omega^n\to\mathbb{R}$ is an orthonormal basis for the coordinate-degree $\s{D}$ functions (with respect to $\left\langle \cdot,\cdot \right\rangle_{\calH_0}$), and that $f_0$ is the unit constant function. Therefore, $\left\langle f_i,f_j \right\rangle_{\calH_0}=\delta_{ij}$, where $\delta_{ij}=0$ if $i\neq j$, and $\delta_{ii}=1$. Then, measuring the norm of $\s{L}_{n,\leq\s{D}}$ in this basis, we have
\begin{align}
\norm{\s{L}_{n,\leq\s{D}}}_{\calH_0}^2 &= \sum_{1\leq i\leq m}\left\langle f_i,\s{L}_{n,\leq\s{D}} \right\rangle_{\calH_0}^2\\
& = \sum_{1\leq i\leq m}\pp{\bE_{\calH_0}\pp{\s{L}_n(\s{G})f_i(\s{G})}}^2\\
& = \sum_{1\leq i\leq m}\pp{\bE_{\calH_1}f_i(\s{G})}^2,\label{eqn:orthogonalDecomposition}
\end{align}
where we have used the fact that $(\s{L}_n-\s{L}_{n,\leq\s{D}})$ is orthogonal to $\{f_i\}_{i=0}^m$. Therefore, we need to compute $\bE_{\calH_1}f_i(\s{G})$ for some orthonormal basis functions $f_i$. In our setting, for $\alpha\subseteq\binom{[n]}{2}$, define the Fourier character
\begin{align}
\chi_{\alpha}(\s{G}) = \prod_{\{i,j\}\in\alpha}\frac{\s{G}_{ij}-q}{\sqrt{q(1-q)}},
\end{align}
for each $\s{G}\in\{0,1\}^{\binom{n}{2}}$. Note that any subset $\alpha\subseteq\binom{[n]}{2}$ can be identified with the graph on vertex set $[n]$ induced by edges in $\alpha$. Therefore, we can say ``graph $\alpha$" without ambiguity. Then, $|\alpha|$ denotes the number of edges in the graph $\alpha$. Also, with some abuse of notation, we let $V(\alpha)$ denote the vertex set of $\alpha$, i.e., the set of vertices $v\in[n]$ that are non-isolated by the edges of $\alpha$, and finally we let $v(\alpha)$ denote its size, namely, $v(\alpha)=|V(\alpha)|$. It is known \cite{odonnell_2014} that $\{\chi_{\alpha}\}_{\alpha\subseteq\binom{[n]}{2},|\alpha|\leq\s{D}}$ form an orthonormal basis for the degree-$\s{D}$ functions with respect to $\pr_{\calH_0}$. 
In light of \eqref{eqn:orthogonalDecomposition}, we next compute $\bE_{\calH_1}\chi_{\alpha}(\s{G})$ for each such $\alpha$. We have,
\begin{align}
    \bE_{\calH_1}\chi_{\alpha}(\s{G}) &= \bE_{\mathsf{K}_{\calR, \calL}}\bE_{\calH_1\vert\mathsf{K}_{\calR, \calL}}\chi_{\alpha}(\s{G})\\
    & = \bE_{\mathsf{K}_{\calR, \calL}}\bE_{\calH_1\vert\mathsf{K}_{\calR, \calL}}\p{\prod_{\{i,j\}\in\alpha}\frac{\s{G}_{ij}-q}{\sqrt{q(1-q)}}}\\
    & = \bE_{\mathsf{K}_{\calR, \calL}}\prod_{\{i,j\}\in\alpha}\bE_{\calH_1\vert\mathsf{K}_{\calR, \calL}}\p{\frac{\s{G}_{ij}-q}{\sqrt{q(1-q)}}},
\end{align}
where we have used the fact that conditioned on the planted structure $\mathsf{K}_{\calR, \calL}$, the edges of $\s{G}$ become independent. Depending on the edge $\{i,j\}$ and its relation to $\alpha$ and $\s{G}$, there are two possible cases:
\begin{itemize}
    \item If $\{i,j\}\in\alpha$ is such that $i,j\not\in V(\mathsf{K}_{\calR, \calL})$, then
\begin{align}
\bE\pp{\left.\frac{\s{G}_{ij}-q}{\sqrt{q(1-q)}}\right|i,j\not\in V(\mathsf{K}_{\calR, \calL})} = 0,
\end{align}
since if $i$ or $j$ is not in $\mathsf{K}_{\calR, \calL}$, then the edge $\{i,j\}$ is included in $\s{G}$ with probability $q$.
\item If $\{i,j\}\in\alpha$ is such that $i,j\in V(\mathsf{K}_{\calR, \calL})$  and $\{i,j\}\in E(\mathsf{K}_{\calR, \calL})$, then
\begin{align}
\bE\pp{\left.\frac{\s{G}_{ij}-q}{\sqrt{q(1-q)}}\right|i,j\in V(\mathsf{K}_{\calR, \calL}),\{i,j\}\in E(\mathsf{K}_{\calR, \calL})} &= p\sqrt{\frac{1-q}{q}}-(1-p)\sqrt{\frac{q}{1-q}}\nonumber\\
& = \sqrt{\chi^2(p||q)}.
\end{align}
\item If $\{i,j\}\in\alpha$ is such that $i,j\in V(\mathsf{K}_{\calR, \calL})$  and $\{i,j\}\not\in E(\mathsf{K}_{\calR, \calL})$, then
\begin{align}
\bE\pp{\left.\frac{\s{G}_{ij}-q}{\sqrt{q(1-q)}}\right|i,j\in V(\mathsf{K}_{\calR, \calL}),\{i,j\}\not\in E(\mathsf{K}_{\calR, \calL})} = 0.
\end{align}
\end{itemize}
Thus, we get
\begin{align}
    \bE_{\calH_1\vert\mathsf{K}_{\calR, \calL}}\p{\frac{\s{G}_{ij}-q}{\sqrt{q(1-q)}}} = \sqrt{\chi^2(p||q)}\cdot\Ind\ppp{\{i,j\}\in\alpha\cap E(\mathsf{K}_{\calR, \calL})}.
\end{align}
This in turn implies that
\begin{align}
    \prod_{\{i,j\}\in\alpha}\bE_{\calH_1\vert\mathsf{K}_{\calR, \calL}}\p{\frac{\s{G}_{ij}-q}{\sqrt{q(1-q)}}} 
    = (\chi^2(p||q))^{\frac{1}{2}|\alpha\cap\mathsf{K}_{\calR, \calL}|}\Ind\ppp{\alpha\subseteq \mathsf{K}_{\calR, \calL}}.
\end{align}
Combining the above we have,
\begin{align}
    \bE_{\calH_1}\chi_{\alpha}(\s{G}) &= \bE_{\mathsf{K}_{\calR, \calL}}\pp{(\chi^2(p||q))^{\frac{1}{2}|\alpha\cap{\mathsf{K}_{\calR, \calL}}|}\Ind\ppp{\alpha\subseteq \mathsf{K}_{\calR, \calL}}}\\
    & = [\chi^2(p||q)]^{\frac{1}{2}|\alpha|}\pr_{\mathsf{K}_{\calR, \calL}}\pp{\alpha\subseteq \mathsf{K}_{\calR, \calL}}.\label{eqn:coeffGenera}
\end{align}
Therefore, from \eqref{eqn:orthogonalDecomposition} we get that
\begin{align}
    \norm{\s{L}_{n,\leq\s{D}}}_{\calH_0}^2 = 
    1+\sum_{0<|\alpha|\leq\s{D}}[\chi^2(p||q)]^{|\alpha|}\pr^2_{\mathsf{K}_{\calR, \calL}}\pp{\alpha\subseteq \mathsf{K}_{\calR, \calL}},\label{eqn:sumoveralphapri} 
\end{align}
where the summand $1$ correspond to the empty graph $\alpha=\emptyset$. Now, the main observation here is that if $\alpha$ is \emph{not} a bipartite graph itself, then $\pr_{\mathsf{K}_{\calR, \calL}}\pp{\alpha\subseteq \mathsf{K}_{\calR, \calL}}=0$. This follows from the fact that any subgraph of a bipartite graph is, itself, bipartite. This can be easily seen as follows: it is well-known that a graph is bipartite if and only if it does not contain a cycle of odd length. Since the subgraph of a bipartite graph could not gain a cycle of odd length, it must also be bipartite. Therefore, the summation in \eqref{eqn:sumoveralphapri} over $\alpha$ can be reduced to bipartite graphs only. Accordingly, let $\calB_{\s{D}}$ denote the set of all possible bipartite subgraphs in the complete graph over $n$ vertices, with at most $\s{D}$ edges, and at least a single edge (because $|\alpha|>0$). For clarity, we will use the notation $\alpha_{\s{Bip}}$ in place of $\alpha$, and we let $v_{\s{L}}(\alpha_{\s{Bip}})$ and $v_{\s{R}}(\alpha_{\s{Bip}})$ denote the number of left and right vertices of $\alpha_{\s{Bip}}$. Thus,
\begin{align}
    \norm{\s{L}_{n,\leq\s{D}}}_{\calH_0}^2 = 
    1+\sum_{\alpha_{\s{Bip}}\in\calB_{\s{D}}}[\chi^2(p||q)]^{|\alpha_{\s{Bip}}|}\pr^2_{\mathsf{K}_{\calR, \calL}}\pp{\alpha_{\s{Bip}}\subseteq \mathsf{K}_{\calR, \calL}}.\label{eqn:sumoveralpha} 
\end{align}
We next evaluate the summation in \eqref{eqn:sumoveralpha}, separately for the dense and sparse regimes, starting with the former. We would like to mention here that, in fact, our analysis for the sparse regime can be generalized so that it will include the dense regime as a special case. Nonetheless, we opted to present the analysis of both cases, as the bounding techniques we use for the dense case are simpler, and might be of use in other related problems.
\subsubsection{Dense regime} 
Consider the case where $p,q = \Theta(1)$, for which $\chi^2(p||q) = \Theta(1)$. Our goal is to prove that if $k_{\s{R}}\vee k_{\s{L}}\ll \sqrt{n}$ then $\norm{\s{L}_{n,\leq\s{D}}}_{\calH_0}^2$ is bounded. For simplicity of notation, we denote $\lambda\triangleq\chi^2(p||q)$, and without loss of generality we assume that $k_{\s{R}}\geq k_{\s{L}}$. Then,
\begin{align}
    \pr_{\mathsf{K}_{\calR, \calL}}\p{\alpha_{\s{Bip}}\subseteq {\mathsf{K}_{\calR, \calL}}} & = \frac{\binom{n-v_{\s{R}}(\alpha_{\s{Bip}})}{k_{\s{R}}-v_{\s{R}}(\alpha_{\s{Bip}})}\binom{n-k_{\s{R}}-v_{\s{L}}(\alpha_{\s{Bip}})}{k_{\s{L}}-v_{\s{L}}(\alpha_{\s{Bip}})}}{\binom{n}{k_{\s{R}}}\binom{n-k_{\s{R}}}{k_{\s{L}}}}\\
    & = \frac{\binom{k_{\s{R}}}{v_{\s{R}}(\alpha_{\s{Bip}})}\binom{k_{\s{L}}}{v_{\s{L}}(\alpha_{\s{Bip}})}}{\binom{n}{v_{\s{R}}(\alpha_{\s{Bip}})}\binom{n-k_{\s{R}}}{v_{\s{L}}(\alpha_{\s{Bip}})}}\label{eqn:coeffGenera2}\\
    & \leq \p{\frac{ek_{\s{R}}}{n}}^{v_{\s{R}}(\alpha_{\s{Bip}})}\p{\frac{ek_{\s{L}}}{n-k_{\s{R}}}}^{v_{\s{L}}(\alpha_{\s{Bip}})}\label{eqn:coeffGenera3}\\
    &\leq \p{\frac{ek_{\s{R}}}{n}}^{v_{\s{R}}(\alpha_{\s{Bip}})}\p{\frac{ek_{\s{R}}}{n-k_{\s{R}}}}^{v_{\s{L}}(\alpha_{\s{Bip}})}\\
    & = \p{\frac{ek_{\s{R}}}{n}}^{v_{\s{R}}(\alpha_{\s{Bip}})+v_{\s{L}}(\alpha_{\s{Bip}})}\p{\frac{n}{n-k_{\s{R}}}}^{v_{\s{L}}(\alpha_{\s{Bip}})}\\
    &\leq \p{\frac{ek_{\s{R}}}{n}}^{v(\alpha_{\s{Bip}})}\exp\p{v_{\s{L}}(\alpha_{\s{Bip}})\cdot\frac{k_{\s{R}}}{n-k_{\s{R}}}},
\end{align}
where the first inequality follows from $(y/x)^x\leq\binom{y}{x}\leq (ey/x)^x$, the second inequality is because $k_{\s{R}}\geq k_{\s{L}}$, and in the last inequality we use $1+x\leq\exp(x)$, for $x\geq0$, and $v(\alpha_{\s{Bip}}) = |V(\alpha_{\s{Bip}})|$ is the total number of non-isolated vertices in $\alpha_{\s{Bip}}$. Now, notice that for any $\alpha_{\s{Bip}}\in\calB_{\s{D}}$, we must have $v_{\s{L}}(\alpha_{\s{Bip}})\leq\s{D}$. Thus, in the regime where $k_{\s{R}}\vee k_{\s{L}}\ll \sqrt{n}$, since $\s{D} = n^{o(1)}$, we have $\exp\p{v_{\s{L}}(\alpha_{\s{Bip}})\cdot\frac{k_{\s{R}}}{n-k_{\s{R}}}} = 1+o(1)$, for sufficiently large $n$. Therefore, we obtain 
\begin{align}
    \pr_{\mathsf{K}_{\calR, \calL}}\p{\alpha_{\s{Bip}}\subseteq {\mathsf{K}_{\calR, \calL}}}\leq(1+o(1))\cdot\p{\frac{ek_{\s{R}}}{n}}^{v(\alpha_{\s{Bip}})}\triangleq g(|V(\alpha_{\s{Bip}})|).
\end{align}
Thus, using \eqref{eqn:sumoveralpha} we get 
\begin{align}
\norm{\s{L}_{n,\leq\s{D}}}_{\calH_0}^2 &\leq 1+\sum_{\alpha_{\s{Bip}}\in\calB_{\s{D}}}\lambda^{|\alpha_{\s{Bip}}|}g^2(|V(\alpha_{\s{Bip}})|).\label{eqn:orthogonalDecomposition2}
\end{align}
It is convenient to analyze the cases $\lambda>1$ and $\lambda\leq1$ separately, and we start with the former. For any set $\calV\subseteq[n]$, we have
\begin{align}
\sum_{\alpha_{\s{Bip}}\in\calB_{\s{D}}:V(\alpha_{\s{Bip}})=\calV}\lambda^{|\alpha_{\s{Bip}}|}g^2(|V(\alpha_{\s{Bip}})|)
& = g^2(|\calV|)\cdot\sum_{\alpha_{\s{Bip}}\in\calB_{\s{D}}:\;v(\alpha_{\s{Bip}})=\calV}\lambda^{|\alpha_{\s{Bip}}|}\\
& \leq g^2(|\calV|)\cdot\sum_{\alpha_{\s{Bip}}\in\calB_{\s{D}}:\;\alpha_{\s{Bip}}\subseteq\binom{\calV}{2}}\lambda^{|\alpha_{\s{Bip}}|}\\
& = g^2(|\calV|)\cdot\sum_{\ell=1}^{\min\p{\s{D},\binom{|\calV|}{2}}}\binom{\binom{|\calV|}{2}}{\ell}\lambda^{\ell}\\
&\leq g^2(|\calV|)\cdot\sum_{\ell=1}^{\min\p{\s{D},\binom{|\calV|}{2}}}\pp{\binom{|\calV|}{2}\lambda}^{\ell}\\
&\leq g^2(|\calV|)\p{1+\binom{|\calV|}{2}\lambda}^{\min\p{\s{D},\binom{|\calV|}{2}}},
\end{align}
where in the last inequality we have used the fact that $\sum_{i=0}^k\binom{n}{i}\leq\sum_{i=0}^k n^i1^{k-i}\leq(1+n)^k$. Next, every $\alpha_{\s{Bip}}$ with $|\alpha_{\s{Bip}}|\leq\s{D}$ has $|v(\alpha_{\s{Bip}})|\leq 2\s{D}$. Also, there are $\binom{n}{t}\leq n^t$ sets $\calV\subseteq[n]$ sets $\calV$ such that $|\calV|=t$. Thus,
\begin{align}
\norm{\s{L}_{n,\leq\s{D}}}_{\calH_0}^2 &\leq1+(1+o(1))\cdot\sum_{t=2}^{2\s{D}}(1+t)^2\p{\frac{ek_{\s{R}}}{n}}^{2t}n^t(1+t^2\lambda/2)^{\min\p{\s{D},t^2/2}}\\
& \leq 1+(1+o(1))\cdot\sum_{t=2}^{2\s{D}}(1+t)^2\p{\frac{e^2k^2_{\s{R}}}{n}}^{t}\pp{t\sqrt{\lambda}}^{2\min\p{\s{D},t^2/2}}\\
&= 1+(1+o(1))\cdot\sum_{t\leq\sqrt{2\s{D}}}(1+t)^2\p{\frac{e^2k^2_{\s{R}}}{n}}^{t}\pp{t\sqrt{\lambda}}^{t^2}\nonumber\\
&\quad+(1+o(1))\cdot\sum_{\sqrt{2\s{D}}<t<2\s{D}}(1+t)^2\p{\frac{e^2k^2_{\s{R}}}{n}}^{t}\pp{t\sqrt{\lambda}}^{2\s{D}},\label{eqn:normL2upper}
\end{align}
where the second inequality follows from the fact that $\lambda>1$ and thus $t^2\lambda>2$. Denote the summand in the first series at the right-hand-side of \eqref{eqn:normL2upper} by $\s{T}_t$. Then, applying the ratio test on the first term we get
\begin{align}
\frac{\s{T}_{t+1}}{\s{T}_t} &= \frac{(2+t)^2}{(1+t)^2}\frac{e^2k^2_{\s{R}}}{n}\frac{(t+1)^{(t+1)^2}}{t^{t^2}}\lambda^{\frac{1}{2}+t}\\
&\leq \frac{(2+t)^2}{(1+t)^2}\frac{e^2k^2_{\s{R}}}{n}e^{t}(1+t)^{1+2t}\lambda^{\frac{1}{2}+t}\\
& \leq (2+t)^{1+2t}\frac{e^2k^2_{\s{R}}}{n}e^{t}\lambda^{\frac{1}{2}+t}\label{eqn:polyMethod0}\\
&\leq (2+\sqrt{2\s{D}})^{1+2\sqrt{2\s{D}}}\frac{e^2k^2_{\s{R}}}{n}e^{\sqrt{2\s{D}}}\lambda^{\frac{1}{2}+\sqrt{2\s{D}}},\label{eqn:polyMethod1}
\end{align}
where in the first inequality we have used the fact that $(1+x)\leq \exp(x)$, for $x\geq0$, the second inequality is because $t\geq 2$, and finally, the last inequality is because $t\leq\sqrt{2\s{D}}$ and the fact that \eqref{eqn:polyMethod0} is monotonically increasing in $t$. Thus, if,
\begin{align}
\frac{e^2k_{\s{R}}^2}{n} e^{\sqrt{2\s{D}}}(2+\sqrt{2\s{D}})^{2\sqrt{2\s{D}}+1}\lambda^{\frac{1}{2}+\sqrt{2\s{D}}}<c<1,\label{eqn:condClique1}
\end{align}
for some $0<c<1$, then the series at the right-hand-side of \eqref{eqn:normL2upper} is upper bounded by a constant. Note that for any $\s{D}\leq C\log n$, and any $C>0$, we have $e^{\sqrt{2\s{D}}}(2+\sqrt{2\s{D}})^{2\sqrt{2\s{D}}+1}\lambda^{\frac{1}{2}+\sqrt{2\s{D}}} = n^{o(1)}$, and thus this factor in \eqref{eqn:condClique1} do not contribute to the polynomial scale of the term in the left-hand-side of \eqref{eqn:condClique1}. Accordingly, it is easy to check that \eqref{eqn:condClique1} holds if $k_{\s{R}} \leq n^{\frac{1}{2}-\epsilon}$, for any $\epsilon>0$. Next, in the same way, we denote the summand in the second term at the right-hand-side of \eqref{eqn:normL2upper} by $\s{T}_t'$. Then, applying the ratio test on the second term we get
\begin{align}
\frac{\s{T}'_{t+1}}{\s{T}'_t} &= \frac{(2+t)^2}{(1+t)^2}\frac{e^2k^2_{\s{R}}}{n}\frac{(1+t)^{2\s{D}}}{t^{2\s{D}}}\\
&\leq \frac{(2+t)^2}{(1+t)^2}\frac{e^2k^2_{\s{R}}}{n}e^{\frac{2\s{D}}{t}}\\
&\leq e^{\sqrt{2\s{D}}}\frac{(2+\sqrt{2\s{D}})^2}{(1+\sqrt{2\s{D}})^2}\frac{e^2k^2_{\s{R}}}{n}.\label{eqn:polyMethod3}
\end{align}
Again, it is easy to check that since $\s{D} = O(\log n)$ then $e^{\sqrt{2\s{D}}}\frac{(2+\sqrt{2\s{D}})^2}{(1+\sqrt{2\s{D}})^2} = n^{o(1)}$, and thus if $k_{\s{R}} \leq n^{\frac{1}{2}-\epsilon}$, for any $\epsilon>0$, then \eqref{eqn:polyMethod3} is strictly less than one, which implies that the second series at the right-hand-side of \eqref{eqn:normL2upper} is upper bounded by a constant. Therefore, for \eqref{eqn:normL2upper} to be bounded as $n\to\infty$ we need $k_{\s{R}} \ll\sqrt{n}$, as claimed. Recall the fact that we have assumed that $k_{\s{R}}\geq k_{\s{L}}$. If reversed, then we will get by symmetry that \eqref{eqn:normL2upper} is bounded as $n\to\infty$ if $k_{\s{L}}\ll\sqrt{n}$; therefore, in general $\norm{\s{L}_{n,\leq\s{D}}}_{\calH_0}^2$ is bounded if $k_{\s{R}}\vee k_{\s{L}}\ll\sqrt{n}$. Finally, we discuss the case where $\lambda\leq1$. In this case, we further bound $\lambda^{|\alpha_{\s{Bip}}|}\leq1$, for any $\alpha_{\s{Bip}}\in\calB_{\s{D}}$, and then \eqref{eqn:orthogonalDecomposition2} becomes,
\begin{align}
\norm{\s{L}_{n,\leq\s{D}}}_{\calH_0}^2 &\leq \sum_{\alpha_{\s{Bip}}\in\calB_{\s{D}}}g^2(|V(\alpha_{\s{Bip}})|)\\
& \leq 1+(1+o(1))\cdot\sum_{t=2}^{2\s{D}}(1+t)^2\p{\frac{e\cdot k_{\s{R}}\vee k_{\s{L}}}{n}}^{2t}n^t.
\end{align}
Using the same arguments as above, we get from the ratio test that $\norm{\s{L}_{n,\leq\s{D}}}_{\calH_0}^2$ is bounded under the same condition, namely, $k_{\s{R}}\vee k_{\s{L}}\ll \sqrt{n}$.

\subsubsection{Sparse regime} 
Next, we analyze the sparse regime where $p,q = \Theta(n^{-\alpha})$, for $\alpha\in(0,2]$. Without loss of generality we assume that $k_{\s{R}}\geq k_{\s{L}}$. Also, recall that $k_{\s{R}},k_{\s{L}}=o(n)$, and thus $k_{\s{R}},k_{\s{L}}\leq n/2$. Finally, we may assume that $\chi^2(p||q)<\frac{1}{2\s{D}^2}$; indeed, in the sparse regime we have $\chi^2(p||q) = \Theta(n^{-\alpha})$, and thus clearly $2\s{D}^2\chi^2(p||q)<1$ for $n$ sufficiently large, and any $\s{D} = n^{o(1)}$. 

Using \eqref{eqn:sumoveralpha} and \eqref{eqn:coeffGenera3}, we have
\begin{align}
    \norm{\s{L}_{n,\leq\s{D}}}_{\calH_0}^2 \leq
    1+\sum_{\alpha_{\s{Bip}}\in\calB_{\s{D}}}[\chi^2(p||q)]^{|\alpha_{\s{Bip}}|}\p{\frac{ek_{\s{R}}}{n}}^{2v_{\s{R}}(\alpha_{\s{Bip}})}\p{\frac{ek_{\s{L}}}{n-k_{\s{R}}}}^{2v_{\s{L}}(\alpha_{\s{Bip}})}.\label{eqn:sumoveralphaSparse} 
\end{align}
Let $\calB_{\ell_1,\ell_2,m}$ denote the set of all edge-induced bipartite subgraphs of the complete graph $\s{K}_n$ on $n$ vertices, having $\ell_1$ right vertices, $\ell_2$ left vertices, and $m$ edges. Formally, 
\begin{align}
    \calB_{\ell_1,\ell_2,m}\triangleq\ppp{\alpha_{\s{Bip}}\subseteq\s{K}_n:v_\s{R}(\alpha_{\s{Bip}})=\ell_1,\;v_\s{L}(\alpha_{\s{Bip}})=\ell_2,\;|\alpha_{\s{Bip}}|=m}.
\end{align}
Note that as the graphs in $\calB_{\ell_1,\ell_2,m}$ are edge induced and we only consider $\alpha_{\s{Bip}}$ such that $|\alpha_{\s{Bip}}|\leq\s{D}$, we have
\begin{align}
    1\leq\ell_1\leq\s{D},\;1\leq\ell_2\leq\s{D},\;\ell_1\vee\ell_2\leq m\leq\s{D}.
\end{align}
Now, using the fact that 
\begin{align}
    \calB_{\s{D}}=\bigcup_{1\leq\ell_1,\ell_2\leq\s{D}}\bigcup_{\ell_1\vee\ell_2\leq m\leq\s{D}}\calB_{\ell_1,\ell_2,m},
\end{align} 
we get
\begin{align}
    \norm{\s{L}_{n,\leq\s{D}}}_{\calH_0}^2 &\leq 1+\sum_{\ell_1,\ell_2=1}^{\s{D}}\sum_{m = \ell_1\vee\ell_2}^{\s{D}}\sum_{\alpha_{\s{Bip}}\in\calB_{\ell_1,\ell_2,m}}(\chi^2(p||q))^{m}\p{\frac{ek_{\s{R}}}{n}}^{2\ell_1}\p{\frac{ek_{\s{L}}}{n-k_{\s{R}}}}^{2\ell_2}\\
    & = 1+\sum_{\ell_1,\ell_2=1}^{\s{D}}\sum_{m = \ell_1\vee\ell_2}^{\s{D}}|\calB_{\ell_1,\ell_2,m}|\cdot(\chi^2(p||q))^{m}\p{\frac{ek_{\s{R}}}{n}}^{2\ell_1}\p{\frac{ek_{\s{L}}}{n-k_{\s{R}}}}^{2\ell_2}.
\end{align}
Let us upper bound $|\calB_{\ell_1,\ell_2,m}|$. To that end, we first enumerate the possible choices for the right and left vertices among the $n$ possible vertices, and then the $m$ edges among the $\ell_1\cdot\ell_2$ possible edges. Specifically,
\begin{align}
    |\calB_{\ell_1,\ell_2,m}|&=\binom{n}{\ell_1}\binom{n-\ell_1}{\ell_2}\binom{\ell_1\cdot\ell_2}{m}\\
    &\leq n^{\ell_1+\ell_2}(\ell_1\cdot\ell_2)^m\\
    &\leq n^{\ell_1+\ell_2}\s{D}^{2m}.
\end{align}
Thus,
\begin{align}
    \norm{\s{L}_{n,\leq\s{D}}}_{\calH_0}^2 &\leq 1+\sum_{\ell_1,\ell_2=1}^{\s{D}}\sum_{m = \ell_1\vee\ell_2}^{\s{D}}(\s{D}^2\chi^2(p||q))^{m}\p{\frac{e^2k^2_{\s{R}}}{n}}^{\ell_1}\p{\frac{e^2k^2_{\s{L}}n}{(n-k_{\s{R}})^2}}^{\ell_2}\\
    &\leq 1+\sum_{\ell_1,\ell_2=1}^{\s{D}}\p{\frac{e^2k^2_{\s{R}}}{n}}^{\ell_1}\p{\frac{4e^2k^2_{\s{L}}}{n}}^{\ell_2}\sum_{m = \ell_1\vee\ell_2}^{\s{D}}(\s{D}^2\chi^2(p||q))^{m}\\
    &\leq 1+2\sum_{\ell_1,\ell_2=1}^{\s{D}}\p{\frac{e^2k^2_{\s{R}}}{n}}^{\ell_1}\p{\frac{4e^2k^2_{\s{L}}}{n}}^{\ell_2}(\s{D}^2\chi^2(p||q))^{\ell_1\vee\ell_2}\\
    & \triangleq1+2\cdot (\calW_{1}+\calW_{2}+\calW_{3}),\label{eqn:sumthree}
\end{align}
where in the second inequality we have used the fact that $k_{\s{R}}\leq n/2$, and the third inequality follows from the inequality $\sum_{i=j}^\infty \beta^i = \frac{\beta^j}{1-\beta}\leq 2\beta^j$, for $\beta<1/2$, along with the fact that $\chi^2(p||q)<\frac{1}{2\s{D}^2}$, and finally in \eqref{eqn:sumthree} we defined,
\begin{align}
    \calW_1&\triangleq \sum_{\ell=1}^{\s{D}}\p{\frac{4e^4k^2_{\s{R}}k^2_{\s{L}}}{n^2}}^{\ell}(\s{D}^2\chi^2(p||q))^{\ell},\label{eqn:W1}\\
    \calW_2&\triangleq \sum_{\ell_1>\ell_2}^{\s{D}}\p{\frac{e^2k^2_{\s{R}}}{n}}^{\ell_1}\p{\frac{4e^2k^2_{\s{L}}}{n}}^{\ell_2}(\s{D}^2\chi^2(p||q))^{\ell_1},\label{eqn:W2}\\
    \calW_3&\triangleq \sum_{\ell_2>\ell_1}^{\s{D}}\p{\frac{e^2k^2_{\s{R}}}{n}}^{\ell_1}\p{\frac{4e^2k^2_{\s{L}}}{n}}^{\ell_2}(\s{D}^2\chi^2(p||q))^{\ell_2}. \label{eqn:W3}   
\end{align}
Next, we find the conditions under which $\calW_1$, $\calW_2$, and $\calW_3$ are all bounded simultaneously, which by using \eqref{eqn:sumthree} will imply that $\norm{\s{L}_{n,\leq\s{D}}}_{\calH_0}^2$ is bounded as well under the same conditions. Starting with $\calW_1$, the geometric series in \eqref{eqn:W1} is bounded provided that $\frac{4\s{D}^2e^4k^2_{\s{R}}k^2_{\s{L}}}{n^2}\chi^2(p||q)<c<1$, for some constant $c$. This is clearly the case when $\chi^2(p||q) \ll \frac{n^2}{k_\s{L}^2k_\s{R}^2}$. Next, we analyze $\calW_2$. We have,
\begin{align}
    \calW_2 = \sum_{\ell_1=1}^{\s{D}}\p{\frac{e^2k^2_{\s{R}}}{n}}^{\ell_1}(\s{D}^2\chi^2(p||q))^{\ell_1}\sum_{\ell_2=1}^{\ell_1}\p{\frac{4e^2k^2_{\s{L}}}{n}}^{\ell_2}.\label{eqn:W2ini}
\end{align}
We consider two cases. If $\frac{4e^2k^2_{\s{L}}}{n}<1$, then we use the fact that $\sum_{i=1}^j\beta^i\leq j\cdot\beta$, for $\beta<1$, and get,
\begin{align}
    \calW_2 &\leq \sum_{\ell_1=1}^{\s{D}}\p{\frac{e^2k^2_{\s{R}}}{n}}^{\ell_1}(\s{D}^2\chi^2(p||q))^{\ell_1}\ell_1\frac{4e^2k^2_{\s{L}}}{n}\\
    &\leq \s{D}\sum_{\ell_1=1}^{\s{D}}\p{\frac{e^2k^2_{\s{R}}}{n}}^{\ell_1}(\s{D}^2\chi^2(p||q))^{\ell_1}.\label{eqn:W20}
\end{align}
Accordingly, the geometric series in \eqref{eqn:W20} is bounded provided that $\frac{\s{D}^2e^2k^2_{\s{R}}}{n}\chi^2(p||q)<c<1$, for some constant $c$. This holds when $\chi^2(p||q) \ll \frac{n}{k^2_\s{R}}$. Conversely, if $\frac{4e^2k^2_{\s{L}}}{n}\geq1$, then we use the fact that $\sum_{i=1}^j\beta^i\leq j\cdot\beta^j$, for $\beta\geq1$, and get,
\begin{align}
    \calW_2 &\leq \sum_{\ell_1=1}^{\s{D}}\p{\frac{e^2k^2_{\s{R}}}{n}}^{\ell_1}(\s{D}^2\chi^2(p||q))^{\ell_1}\ell_1\p{\frac{4e^2k^2_{\s{L}}}{n}}^{\ell_1}\\
    &\leq \s{D}\sum_{\ell_1=1}^{\s{D}}\p{\frac{4e^4k^2_{\s{R}}k^2_{\s{L}}}{n^2}}^{\ell_1}(\s{D}^2\chi^2(p||q))^{\ell_1}.\label{eqn:W21}
\end{align}
As for $\calW_1$, the series in \eqref{eqn:W21} is bounded if $\chi^2(p||q) \ll \frac{n^2}{k_\s{L}^2k_\s{R}^2}$. Finally, we analyze $\calW_3$ exactly as we analyzed $\calW_2$. We get that if $\frac{e^2k^2_{\s{R}}}{n}<1$ then $\calW_3$ is bounded provided that $\frac{4\s{D}^2e^2k^2_{\s{L}}}{n}\chi^2(p||q)<c<1$, for some constant $c$. This is clearly the case when $\chi^2(p||q) \ll \frac{n}{k^2_\s{L}}$; since we assume that $k_\s{R}\geq k_\s{L}$ then this condition is redundant because $\frac{e^2k^2_{\s{R}}}{n}<1$. Conversely, if $\frac{e^2k^2_{\s{R}}}{n}\geq1$, we get that $\calW_3$ is bounded provided that $\chi^2(p||q) \ll \frac{n^2}{k_\s{L}^2k_\s{R}^2}$. Again, since $k_\s{R}\geq k_\s{L}$ then this condition is less stringent compared to the conditions we got for $\calW_2$. Combining the above, we get that $\calW_1$, $\calW_2$, and $\calW_3$, are all bounded when:
\begin{enumerate}
    \item If $k^2_{\s{L}}\ll n$, then we need $\chi^2(p||q) \ll \frac{n^2}{k_\s{L}^2k_\s{R}^2}$ and $\chi^2(p||q) \ll \frac{n}{k^2_\s{R}}$, and since the later dominates, the intersection reduces to $\chi^2(p||q) \ll \frac{n}{k^2_\s{R}}$. 
    \item If $k^2_{\s{L}}\gg n$, then we need $\chi^2(p||q) \ll \frac{n^2}{k_\s{L}^2k_\s{R}^2}$.
\end{enumerate}
Notice that when the above two cases are combined we get the condition $\chi^2(p||q) \ll \frac{n^2}{k_\s{L}^2k_\s{R}^2}\wedge\frac{n}{k^2_\s{R}}$, and if we remove the assumption that $k_\s{R}\geq k_\s{L}$, we obtain
$\chi^2(p||q) \ll \frac{n^2}{k_\s{L}^2k_\s{R}^2}\wedge\frac{n}{k^2_\s{R}\vee k^2_\s{L}}$. This concludes the proof.

\subsection{Converse}

In this subsection, we prove the converse part of Theorem~\ref{thm:gap}, namely, that in the dense and sparse regimes if $k^2_{\s{R}}\vee k^2_{\s{L}} \gg n$ and $\chi^2(p||q)\gg \frac{n^2}{k^2_{\s{R}} k^2_{\s{L}}}\wedge\frac{n}{k^2_{\s{R}}\vee k^2_{\s{L}}}$ then $\norm{\s{L}_{n,\leq \s{D}}}_{\calH_0}\geq \omega(1)$, respectively. Recall that we assume that $k_{\s{R}},k_{\s{L}}=o(n)$, and thus $k_{\s{R}},k_{\s{L}}\leq n/2$. Next, we lower bound $\norm{\s{L}_{n,\leq\s{D}}}_{\calH_0}^2$ in \eqref{eqn:sumoveralpha}. Using \eqref{eqn:coeffGenera2} and the inequality $(y/x)^x\leq\binom{y}{x}\leq (ey/x)^x$, we have,
\begin{align}
    \pr_{\mathsf{K}_{\calR, \calL}}\pp{\alpha_{\s{Bip}}\subseteq \mathsf{K}_{\calR, \calL}}&=\frac{\binom{k_{\s{R}}}{v_{\s{R}}(\alpha_{\s{Bip}})}\binom{k_{\s{L}}}{v_{\s{L}}(\alpha_{\s{Bip}})}}{\binom{n}{v_{\s{R}}(\alpha_{\s{Bip}})}\binom{n-k_{\s{R}}}{v_{\s{L}}(\alpha_{\s{Bip}})}}\\
    &\geq \p{\frac{k_{\s{R}}}{en}}^{v_{\s{R}}(\alpha_{\s{Bip}})}\p{\frac{k_{\s{L}}}{e(n-k_{\s{R}})}}^{v_{\s{L}}(\alpha_{\s{Bip}})},
\end{align}
for any $\alpha_{\s{Bip}}\in\calB_{\s{D}}$. Then, using similar arguments that lead to \eqref{eqn:sumthree}, we get
\begin{align}
    \norm{\s{L}_{n,\leq\s{D}}}_{\calH_0}^2 &\geq \sum_{\ell_1,\ell_2\geq1}\sum_{m\geq \ell_1\vee\ell_2}|\calB_{\ell_1,\ell_2,m}|\cdot(\chi^2(p||q))^{m}\p{\frac{k_{\s{R}}}{en}}^{2\ell_1}\p{\frac{k_{\s{L}}}{e(n-k_{\s{R}})}}^{2\ell_2}\\
    &\geq \sum_{\ell_1,\ell_2\geq1}\sum_{m\geq \ell_1\vee\ell_2}\p{\frac{n}{e\ell_1}}^{\ell_1}\p{\frac{n-\ell_1}{e\ell_2}}^{\ell_2}(\chi^2(p||q))^{m}\p{\frac{k_{\s{R}}}{en}}^{2\ell_1}\p{\frac{k_{\s{L}}}{e(n-k_{\s{R}})}}^{2\ell_2}\\
    & \geq \sum_{\ell_1,\ell_2\geq1}(\chi^2(p||q))^{\ell_1\vee\ell_2}\p{\frac{k^2_{\s{R}}}{e^3n\ell_1}}^{\ell_1}\p{\frac{k^2_{\s{L}}(n-\ell_1)}{e^3(n-k_{\s{R}})^2\ell_2}}^{\ell_2}\\
    & = \tilde\calW_1+\tilde\calW_2+\tilde\calW_3\\
    &\geq\tilde\calW_1\vee\tilde\calW_2\vee\tilde\calW_3,\label{eqn:tildeW}
\end{align}
where in the second inequality have used the fact that $|\calB_{\ell_1,\ell_2,m}|\geq \binom{n}{\ell_1} \binom{n-\ell_1}{\ell_2}\geq \p{\frac{n}{e\ell_1}}^{\ell_1}\p{\frac{n-\ell_1}{e\ell_2}}^{\ell_2}$, for any $m$, and we defined,
\begin{align}
    \tilde\calW_1&=\sum_{\ell=1}^{\s{D}}(\chi^2(p||q))^{\ell}\p{\frac{k^2_{\s{R}}}{e^3n\ell}}^{\ell}\p{\frac{k^2_{\s{L}}(n-\ell)}{e^3(n-k_{\s{R}})^2\ell}}^{\ell},\label{eqn:tideW1}\\
    \tilde\calW_2&=\sum_{\ell_1>\ell_2}^{\s{D}}(\chi^2(p||q))^{\ell_1}\p{\frac{k^2_{\s{R}}}{e^3n\ell_1}}^{\ell_1}\p{\frac{k^2_{\s{L}}(n-\ell_1)}{e^3(n-k_{\s{R}})^2\ell_2}}^{\ell_2},\label{eqn:tideW2}\\
    \tilde\calW_3&=\sum_{\ell_2>\ell_1}^{\s{D}}(\chi^2(p||q))^{\ell_2}\p{\frac{k^2_{\s{R}}}{e^3n\ell_1}}^{\ell_1}\p{\frac{k^2_{\s{L}}(n-\ell_1)}{e^3(n-k_{\s{R}})^2\ell_2}}^{\ell_2}.\label{eqn:tideW3}
\end{align}
To prove the converse, it turns out that it is suffice to find conditions under which either $\tilde\calW_1$ or $\tilde\calW_2$ at the right-hand-side of \eqref{eqn:tildeW} diverge, which in turn imply that $\norm{\s{L}_{n,\leq\s{D}}}_{\calH_0}^2$ diverge as well under the same conditions. Starting with $\tilde\calW_1$, by lower bounding $\tilde\calW_1$ by the first term in the sum (i.e., the one that corresponds to the index $\ell=1)$, we get $\tilde\calW_1\geq \chi^2(p||q)\frac{k^2_{\s{R}}}{e^3n}\frac{k^2_{\s{L}}(n-1)}{e^3(n-k_{\s{R}})^2}\geq \chi^2(p||q)\frac{k^2_{\s{R}}}{e^3n}\frac{k^2_{\s{L}}}{e^3n}$, where the second inequality follows from the fact that $\frac{n-\ell_1}{(n-k_{\s{R}})^2}\geq\frac{1}{n}$, because $\ell_1<n$. Thus, if $\chi^2(p||q)\gg\frac{n^2}{k_{\s{L}}^2k_{\s{R}}^2}$, then $\tilde\calW_1=\omega(1)$, which implies $\norm{\s{L}_{n,\leq\s{D}}}_{\calH_0}^2=\omega(1)$. Next, we prove that $\norm{\s{L}_{n,\leq\s{D}}}_{\calH_0}^2=\omega(1)$ also if $\chi^2(p||q)\gg\frac{n}{k_{\s{R}}^2\vee k_{\s{L}}^2}$, by analyzing $\tilde\calW_2$. We have,
\begin{align}
    \tilde\calW_2 &= \sum_{\ell_1=1}^{\s{D}}(\chi^2(p||q))^{\ell_1}\p{\frac{k^2_{\s{R}}}{e^3n\ell_1}}^{\ell_1}\sum_{\ell_2=1}^{\ell_1}\p{\frac{k^2_{\s{L}}(n-\ell_1)}{e^3(n-k_{\s{R}})^2\ell_2}}^{\ell_2}\\
    &\geq\sum_{\ell_1=1}^{\s{D}}(\chi^2(p||q))^{\ell_1}\p{\frac{k^2_{\s{R}}}{e^3n\ell_1}}^{\ell_1}\p{\frac{k^2_{\s{L}}(n-\ell_1)}{e^3(n-k_{\s{R}})^2}}\\
    &\geq \frac{k^2_{\s{L}}}{e^3n}\sum_{\ell_1=1}^{\s{D}}(\chi^2(p||q))^{\ell_1}\p{\frac{k^2_{\s{R}}}{e^3n\ell_1}}^{\ell_1}\\
    & = \frac{k^2_{\s{L}}}{e^3n}\frac{\zeta\cdot(\zeta^{\s{D}}-1)}{\zeta-1},\label{eqn:tildeW2ini}
\end{align}
where $\zeta\triangleq\chi^2(p||q)\frac{k^2_{\s{R}}}{\s{D}e^3n}$. We claim that there exist an integer $\s{D}$ such that if $\zeta\gg1$ (which is equivalent to $\chi^2(p||q)\gg\frac{n}{k^2_{\s{R}}}$) then $\tilde\calW_2=\omega(1)$. First, for $\zeta>2^{1/\s{D}}$ note that \eqref{eqn:tildeW2ini} is further lower bounded by $\frac{k^2_{\s{L}}}{2e^3n}\zeta^{\s{D}}$. Now, if $k^2_{\s{L}}\gg n$, then clearly \eqref{eqn:tildeW2ini} diverges when $\zeta\gg1$. If $k^2_{\s{L}}\ll n$, then since the ratio $\frac{k^2_{\s{L}}}{n}$ decreases at most polynomially with $n$, while $\zeta$ increases polynomially with $n$, then there must be an integer $\s{D}$ such that the product $\frac{k^2_{\s{L}}}{2e^3n}\zeta^{\s{D}}$ diverges. Let $\kappa_n$ be any sequence such that $\kappa_n=\omega(1)$. Then, any $\s{D}$ such that \begin{align}
    \s{D}\geq\frac{\log\kappa_n}{\log\xi}+\frac{\log\frac{2e^3n}{k^2_{\s{L}}}}{\log\xi},
\end{align}
implies that $\tilde\calW_2=\omega(1)$. If $\xi\gg1$, we have $\xi = \Theta(n^{\epsilon})$, for some $\epsilon>0$. Also, since $k^2_{\s{L}}\ll n$, then $\frac{2e^3n}{k^2_{\s{L}}} = \Theta(n^{\eta})$, from some $\eta>0$. Thus, for any $\kappa_n = o(\log n)$, if we take any integer $\s{D}$ such that
\begin{align}
    \s{D}\geq \frac{\eta}{\epsilon}+o(1),
\end{align}
then $\tilde\calW_2=\omega(1)$. 
Finally, by replacing the roles of $k_{\s{R}}$ and $k_{\s{L}}$ we obtain the condition $\chi^2(p||q)\gg\frac{n}{k^2_{\s{L}}}$, and thus when combined with the sufficient condition we obtained for $\tilde\calW_1$, we conclude that $ \norm{\s{L}_{n,\leq\s{D}}}_{\calH_0}^2=\omega(1)$ if $\chi^2(p||q)\gg\frac{n^2}{k_{\s{L}}^2k_{\s{R}}^2}\wedge\frac{n}{k^2_{\s{R}}\vee k^2_{\s{L}}}$, as stated.

\section{Conclusion} \label{sec:conclusion}
In this work, we studied the problem of detecting a planted bipartite structure in a large random graph. We derived asymptotically tight lower and upper bounds in both the dense and sparse regimes. It turns out that the, seemingly simple, bipartite structure has essentially a rich geometry, which produces a surprising phase diagram for the detection problem, as can be seen in, for example, Figure~\ref{fig:planted bipartite}. Finally, we provided an evidence that in the region of parameters where the inefficient scan test is successful, but the count and degree tests fail, the problem is conjecturally computationally hard, by proving that the class of low-degree polynomials fail in this region. An interesting direction for future research is to consider the more general case where $p = \Theta(n^{-\beta})$ and $q = \Theta(n^{-\gamma})$, for $\beta<\gamma$. It seems that while some of our results and techniques apply for this case as well, in order to obtain tight statistical and computational barriers, other techniques will be needed, especially in the regime where $k_{\s{R}}\vee k_{\s{R}}\ll\sqrt{n}$ and $\gamma>2\beta$, for which $\chi^2(p||q) = \Theta(n^{\gamma-2\beta})\gg1$. Finally, while in this paper we analyzed the detection problem, it is quite interesting and important to consider the recovery problem as well. Our detection algorithms can be easily converted to a respective recovery algorithm (e.g., the scan test is motivated by the maximum likelihood estimator). For lower bounds, other techniques are needed. For example, for computational lower bounds it will be interesting to use the recent framework of low-degree polynomials for recovery \cite{SCHRAMM22}. 

\bibliographystyle{unsrt}
\bibliography{bibfile,bibfile2}

\end{document}